\documentclass[11pt]{article} % For LaTeX2e
\usepackage{amsmath,amssymb,amsthm,mathtools}
\usepackage[T1]{fontenc}

\usepackage{newtxtext}
\usepackage[section]{placeins}
\usepackage{float}
\usepackage{newtxmath}
\usepackage[margin=1in]{geometry}
\usepackage{booktabs}
\usepackage{algorithm}
\usepackage{multirow}
\usepackage{algpseudocode}
\usepackage{enumitem}
\usepackage{comment}
\usepackage{subcaption}
\usepackage{color,soul}
\usepackage{enumitem}
\usepackage{siunitx}
\usepackage{amsmath,amsfonts,bm}

\def\eqref#1{equation~\ref{#1}}
\def\1{\bm{1}}

\DeclareMathAlphabet{\mathsfit}{\encodingdefault}{\sfdefault}{m}{sl}
\SetMathAlphabet{\mathsfit}{bold}{\encodingdefault}{\sfdefault}{bx}{n}

\newcommand{\E}{\mathbb{E}}

\DeclareMathOperator*{\argmin}{arg\,min}

\usepackage[backref,colorlinks,linkcolor=blue,citecolor=blue,bookmarks=true]{hyperref}
\usepackage{mathtools, amssymb, amsthm, dsfont}
\numberwithin{equation}{section}
\usepackage{enumitem}
\usepackage[nameinlink,capitalize]{cleveref}
\usepackage{algorithm}
\usepackage[table]{xcolor}
\usepackage{natbib}
\usepackage{setspace}

\usepackage{amsmath, amsfonts, amsthm,amssymb}
\usepackage{verbatim, comment}
\usepackage{float}
\usepackage{bbm}
\usepackage{enumitem,graphicx,psfrag,epsf,placeins,subcaption}
\usepackage{algorithm}
\usepackage{algpseudocode}
\usepackage{multirow}
\usepackage{bm}
\usepackage{tabularx}
\usepackage{tikz} %for graphs
\usetikzlibrary{trees}
\usepackage{mathtools} %for \coloneq
\usepackage{thmtools}
\usepackage{thm-restate}
\usepackage{amssymb}

\usepackage{hyperref}
\usepackage{cleveref}
\Crefname{equation}{Eq.}{Eqs.}
\Crefname{figure}{Fig.}{Figs.}
\usepackage{enumitem}

\usepackage{booktabs}
\usepackage{longtable}
\usepackage{array}
\usepackage{multirow}
\usepackage{wrapfig}
\usepackage{float}
\usepackage{colortbl}
\usepackage{pdflscape}
\usepackage{tabu}
\usepackage{threeparttable}
\usepackage{threeparttablex}
\usepackage[normalem]{ulem}
\usepackage{makecell}
\usepackage{color}

\newtheorem{lemma}{Lemma}

\newtheorem{assumption}{Assumption}

\numberwithin{equation}{section}

\newcommand{\var}{{\rm var}}       % for variance
\newcommand{\Istar}{\mathcal{I}^\star}
\newcommand{\I}{\mathcal{I}}
\newcommand{\gammadag}{\bm \gamma^\dagger}
\newcommand{\leaves}{\rm le}
\newcommand{\sib}{{\rm sib}}
\newcommand{\an}{{\rm an}}
\newcommand{\de}{{\rm de}}
\newcommand{\pa}{{\rm pa}}
\newcommand{\g}{\bm \gamma}
\newcommand{\X}{\mathbf{X}}
\newcommand{\Y}{\mathbf{Y}}
\newcommand{\tilX}{\mathbf{\tilde  X}}
\newcommand{\T}{\mathcal{T}}
\newcommand{\K}{K_{\g}}

\usepackage{tikz}
\usetikzlibrary{shapes,arrows,positioning}

\usepackage{parskip}
\usepackage{mathtools, amssymb, amsthm, bbm}
\usepackage[nameinlink,capitalize]{cleveref}

\theoremstyle{plain}

\theoremstyle{definition}

\theoremstyle{remark}

\usepackage{url}
\usepackage{etoc}
\usepackage[title]{appendix}

\definecolor{light-gray}{gray}{0.95}

\definecolor{HHgreen}{HTML}{34C759}

\newlist{thmassumptions}{enumerate}{1}
\setlist[thmassumptions,1]{label=(\alph*), ref=\thetheorem(\alph*)}

\crefname{thmassumptionsi}{Assumption}{assumptions}
\crefname{thmassumptionsi}{Assumption}{Assumptions}

\title{Nonparametric Regression via Tree-Guided Feature
Aggregation}

\author{Sithija Manage\thanks{
    The authors gratefully acknowledge \textit{the National Institutes of Health grant T32HD113301: Artificial Intelligence and Precision Nutrition Training Program, Cornell University. }}\hspace{.2cm}\\
    Department of Statistics and Data Science, Cornell University\\
    and \\
    Y. Samuel Wang \\
    Department of Statistics and Data Science, Cornell University \\
    and \\
    Martin T. Wells \\
    Department of Statistics and Data Science, Cornell University}


\begin{document}

\maketitle
\begin{abstract}
In regression problems where covariates are naturally organized in a hierarchical tree structure, a central challenge is to select the resolution at which covariates enter the model. Determining this level of feature aggregation is of intrinsic scientific interest and can improve statistical efficiency by inducing sparsity.
While a rich literature addresses this problem in the linear setting, extending feature aggregation to the nonlinear setting remains an open challenge. In this work, we propose to simultaneously perform model selection and feature aggregation through a penalized Nadaraya-Watson-type estimator. Our proposed estimator, Kernel Regression with Tree‑EXploring AggregationS (KR-TEXAS), constructs adaptive penalty weights for the features based on pilot estimators of the regression function's partial derivatives. Under mild conditions, we establish model selection consistency for a well-defined target aggregation set, and our simulations show strong performance in both model selection and prediction. Finally, we demonstrate the utility of our procedure by applying it to a microbiome data set to predict short chain fatty acids. 
A user-friendly implementation of our procedure is available in the \texttt{R} package \texttt{krtexas}.
\end{abstract}

\section{Introduction}\label{sec:intro}
Many modern regression problems involve nonlinear relationships with covariates that are naturally organized in a hierarchical tree structure. A central statistical challenge in this setting is selecting the resolution at which covariates enter the model. One may include covariates at the individual level, or alternatively form aggregated variables by summing covariates according to the tree topology. The appropriate resolution is of intrinsic scientific interest, as it reveals the scale of granularity at which the underlying process operates, and thereby aids in interpretability. Moreover, aggregation can improve statistical efficiency by inducing sparsity, an advantage that is especially pronounced in the nonparametric regression framework that we develop here.

Such regression problems arise naturally across diverse scientific domains. In neuroimaging, for example, MRI data may be used to predict cognitive decline in Alzheimer's patients, where measurements at the individual voxel level may be used or aggregated by brain sub-region, with sub-regions themselves admitting further aggregation \citep{wang2022regularized}. In microbiome analysis, the covariates are abundances of operational taxonomic units (OTUs), and one may either use individual OTU counts directly or aggregate them at progressively coarser taxonomic levels according to the phylogenetic tree~\citep{wang2017constructing, wang2017structured, lee2025zero}. %Similar hierarchical structure exists in modeling product 
%, tree-guided perceptual mapping in marketing \citep{bendixen1995compositional}, financial risk allocation and stock market analysis structured by economic region \citep{fiori2025compositional, mantegna1999hierarchical},  
Psychologists have used spending records from different retailers to predict personality traits~\citep{gladstone2019traits}. Each of these merchants could also be categorized in a hierarchical tree; e.g., purchases at McDonald's could be included as a covariate or could be included at different levels of aggregation such as total fast food spending, total restaurant spending, or total food spending.

The appropriate level of aggregation may differ between covariate groups, motivating data-driven procedures for joint model selection across the tree. While various such procedures have been proposed for linear regression, we focus on the case in which the outcome is a nonlinear function of the covariates, whose relevant aggregation structure is unknown.

\subsection{Contribution}
As discussed below, most methods in the linear setting aggregate features by fusing together regression coefficients. However, extending feature aggregation to the nonlinear regression setting requires a very different approach. Specifically, we propose to simultaneously perform model selection and feature aggregation in nonparametric regression through a penalized Nadaraya-Watson type estimator. Our key innovation is the careful construction of adaptive penalty weights which use pilot estimators of the partial derivatives. 

Under mild conditions, we show that our procedure is model selection consistent for a targeted aggregation set. In simulations, we see that our procedure has good empirical performance in both model selection and prediction performance, even showing performance comparable to certain oracle methods. Finally, we demonstrate the utility of our procedure by applying it to a microbiome data set to predict short chain fatty acids.

\subsection{Previous work}
Various approaches have been proposed for linear regression with tree-structured covariates. In the linear setting, groups of covariates can be aggregated by forcing them to share a common regression coefficient. 
\cite{wang2017constructing} consider compositional covariates and propose a tree-guided fused lasso~\citep{tibshirani2005sparsity} penalty which encourages covariates in the same subtree to have identical coefficients. Closer to our work, \cite{bien2021rare} propose an overparameterized model where the coefficients for each observed covariate are a sum of terms corresponding to its ancestors in the hierarchical tree. An $L_1$ penalty for these terms acts similarly to a fused lasso, encouraging the same coefficient for covariates that are in the same sub-tree. \cite{wang2022regularized} use a similar parameterization, but modify the penalization and include additional constraints to handle compositional data.
These tree-structured regression methods have also been extended to regression with ``or'' operators over binary variables~\citep{chen2024tree}, linear regression with multivariate outcomes~\citep{mishra2024taro}, and Gaussian graphical models~\citep{wilms2022tree}. \cite{fu2025direct} recently proposed a modified penalization that avoids overparameterization.

In order to generalize these methods to a nonparametric setting, we use a regularized Nadaraya-Watson~\citep{nadaraya1964estimating, watson1964smooth} estimator which optimizes a leave-one-out criterion over possibly anisotropic bandwidths. These types of procedures may adapt to the sparsity in observed covariates~\citep{hall2007nonparametric, conn2019oracle} and are referred to as metric learning in the machine learning community~\citep{weinberger2007metric, noh2017metric}. In contrast to these previous works, our main focus is on model selection, not simply improved predictions. Thus, our method is most similar to~\cite{white2017variable} which also casts variable selection as a regularized bandwidth selection problem. However, because our model also includes aggregated variables, new techniques are required.

Our problem is also reminiscent of multi-index models~\citep{feng2013partial, xia2008multiple, yang2017learning} and sufficient dimension reduction~\citep{li1991sliced, ma2013review, globerson2003sufficient} which similarly posit that the true regression function depends only on a lower-dimensional linear transformation of the original features. However, in our setting, the possible transformations are constrained by the known hierarchical tree.

The remainder of this paper is organized as follows: Section \ref{sec:methods} describes and provides the mathematical formulation of KR-TEXAS, Section \ref{sec:theory} establishes theoretical guarantees, Section \ref{sec:numerics} details our numerical experiments, and Section \ref{sec:analysis} uses the proposed method to identify relevant microbial features to predict SCFA levels.

\section{Methodology}\label{sec:methods}

\subsection{Model}
\label{sec:methods:model}
We observe i.i.d. data \( (X_i, Y_i) \) for \( i = 1, \ldots, n \), where \( X_i = (X_{ij} \, : \, j = 1, \ldots, p) \in \mathcal{X}\) denotes the covariate vector for the \( i \)-th sample, and \( Y_i \in \mathbb{R} \) is the response. We assume $Y_i$ is related to $X_i$ through a nonparametric regression function:
\begin{equation}\label{eq:np_regLeaves}
Y_i = m(X_i) + \varepsilon_i,\end{equation}
where $\E(\varepsilon_i) = 0$, $\var(\varepsilon_i) = \sigma^2$ and $\varepsilon_i$ are independent of $X_i$. Let \( \X \in \mathbb{R}^{n \times p}\) denote the matrix whose $i$th row is $X_i$, $\Y = (Y_i\,:\, i = 1, \ldots, n)$, $\X_{\cdot, j}$ denote the $j$th column of $\X$, and $\X_{j, \cdot}$ denote the $j$th row of $\X$. Similarly, for a set $J$, let $\X_{\cdot, J}$ and $\X_{J, \cdot}$ denote the sub-matrix of $\X$ with columns/rows corresponding to $J$. We will also use $[n] := \{1, 2, \ldots, n\}$.

We further assume that the covariates $X_i$, are hierarchically structured according to a known rooted tree $\mathcal{T} = \{V, E\}$ where $V$ are the vertices with $T = \vert V\vert$ and $E$ is the set of directed edges. 
For a node $v \in \mathcal{T}$, the parent of $v$ is $\pa(v) = \{ u \in V \,:\, u \rightarrow v \in E\}$, the ancestors of $v$ are $\an(v) = \{u \in V \,:\, \exists \text{ a directed path from } u \text{ to } v \}$, the descendants of $v$ are $\de(v) = \{u \in V \,:\, \exists \text{ a directed path from } v \text{ to } u \}$, and the siblings of $v$ are $\sib(v) = \{u \in V \,:\, \pa(v) = \pa(u)\}$. When applied to a set, the notation should be read as the union of relevant sets; e.g., for $J \subset V$, $\pa(J) = \cup_{j \in J} \pa(j)$.  A node in $\T$ is a leaf node if it does not have descendants. We will use $\leaves(v)$ to denote the subset of $\de(v) \cup v$ which are leaves. The tree $\T$ has $p$ leaf nodes that each correspond to an observed feature; i.e., an element of $X_i$. Each internal node (i.e., non-leaf) represents a variable which is the sum of its descendant leaves. The hierarchical structure may be encoded in a binary matrix $ \bm A \in \{0,1\}^{T \times p} $ where each row corresponds to a node in the tree and each column corresponds to a leaf in the tree. We then set $A_{uv} =1$ if $v \in \leaves(u)$ and define the aggregated variables as
\begin{equation}
 \tilde X_i = \bm{A}X_i.   
\end{equation}
Thus, $\tilde X_i \in \mathbb{R}^T$ contains each of the $p$ originally observed variables (which correspond to leaves in $\T$) as well as the variables corresponding to the internal nodes $(T-p)$ such that $\tilde X_{i,v} = \sum_{u \in \leaves(v)} \X_{i,u}$. We will use $\tilX = \X \mathbf{A}^T$ to denote the $n \times T$ matrix where each row corresponds to an aggregated observation $\tilde X_i$. For example, in a gut microbiome analysis, $X_i$ may represent subject $i$'s species-level microbial abundance data, with each of the $p$ values corresponding to a relative abundance in $[0,1].$ In this example $\mathbf{A}$ may encode the taxonomic tree structure, informing us how the bacterial species are organized in a phylogenetic tree. For an internal node $v$ corresponding to a genus, $\tilde X_{i,v}$ is the sum of relative abundances for all species within genus $v$; i.e., the relative abundance of genus $v$. Thus,  $\tilde X_{i}$ contains subject $i$'s aggregated abundances at the genus, family, class, and higher taxonomic levels, along with the original relative abundances of the $p$ species.

In Fig.~\ref{fig:treeExample}, we show an example tree and the corresponding $\mathbf{A}$. There are $5$ observed features that correspond to the vertices $1$ through $5$ and the first $5$ rows of $A$. There are $3$ additional nodes that correspond to the aggregated variables. For example, node $6$ represents the aggregation of $4$ and $5$ so $\tilde \X_{\cdot, 6} = \X_{\cdot, 4} + \X_{\cdot, 5}$; similarly, node $7$ represents the aggregation $1, 4, 5$ so $\tilde \X_{\cdot, 7} = \X_{\cdot, 1} +  \X_{\cdot, 4} + \X_{\cdot, 5}$. In a slight abuse of notation, we will often use $v \in V$ to refer to both the node in $\T$ and its corresponding variable in $\tilde X_i$. 

\begin{figure}[h]
    \centering
    \begin{minipage}{0.2\textwidth}
        \begin{tikzpicture}[level distance=1.2cm,
          level 1/.style={sibling distance=1.2cm},
          level 2/.style={sibling distance=1.2cm},
          level 3/.style={sibling distance=1.2cm},
          every node/.style={circle, draw, minimum size=1cm, inner sep=0pt, align=center},
          edge from parent/.style={draw, ->, >=stealth}
          ]
          \node {\textcolor{blue}{$\tilde \X_{\cdot, 8}$}}
            child {node {\textcolor{blue}{$\tilde \X_{\cdot, 7}$}}
              child {node {$\tilde \X_{\cdot, 1}$}
              }
              child {node {\textcolor{blue}{$\tilde \X_{\cdot, 6}$}}
                child {node {$\tilde \X_{\cdot, 4}$}}
                child {node {$\tilde \X_{\cdot, 5}$}}
              }
            }
            child {node {$\tilde \X_{\cdot, 2}$}
            }
            child {node {$\tilde \X_{\cdot, 3}$}
            };
        \end{tikzpicture}
    \end{minipage}
    \hspace{1cm} % Adjust the space between the tree and the matrix
    \begin{minipage}{0.25\textwidth}
        \[
        \bm A =
        \begin{bmatrix}
        1 & 0 & 0 & 0 & 0 \\
        0 & 1 & 0 & 0 & 0 \\
        0 & 0 & 1 & 0 & 0 \\
        0 & 0 & 0 & 1 & 0 \\
        0 & 0 & 0 & 0 & 1 \\
        0 & 0 & 0 & \textcolor{blue}{1} & \textcolor{blue}{1} \\
        \textcolor{blue}{1} & 0 & 0 & \textcolor{blue}{1} & \textcolor{blue}{1} \\
        \textcolor{blue}{1} & \textcolor{blue}{1} & \textcolor{blue}{1} & \textcolor{blue}{1} & \textcolor{blue}{1} \\
        \end{bmatrix}
        \]
    \end{minipage}
    \caption{Example tree with 5 leaves, 3 nodes that correspond to the aggregated variables (blue), and corresponding $A$ matrix.}
    \label{fig:treeExample}
\end{figure}
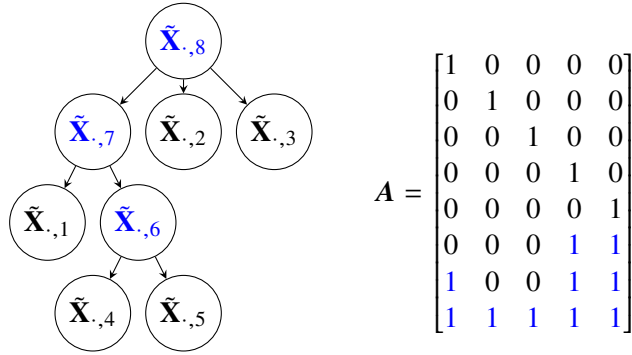

In a typical setting, one might assume that the model is sparse in the observed features so that $m$ (and subsequently $Y$) only depends on a small subset, $\I \subset [p]$, of the elements of $X_i$; i.e., for some function $m_{\I}$:
\begin{equation}\label{eq:np_regProjLeaves}
m(X_i) = m_{\I}\left(P_{\I}(X_i)\right)
\end{equation}
where $P_{\I}(X_i)$ is the coordinate projection that returns the elements of $X_i$ corresponding to the set of indices $\I$. We assume that the model may also be sparse in the aggregated variables such that for some set $\mathcal{I} \subset [T]$:
\begin{equation}\label{eq:np_regProj}
m(X_i) = \tilde m_{\I}\left(P_{\I}(\tilde X_i)\right).
\end{equation}

%We will also use $P_{\I}(\X)$ to denote the coordinate projection which returns the columns of $\X$ corresponding to $\I$. 

Our primary goal will be to identify a small set $\I$ for which Eq.~\eqref{eq:np_regProj} holds. This is analogous to the original goal of~\cite{bien2021rare} who consider the linear model setting. In many applications, the set $\I$ is of intrinsic interest and aids in interpretability by indicating the appropriate level of data granularity for the regression task at hand. For example, according to phylogenetic niche conservatism theory \citep{yu2020new}, microbial species diverging from the same clade may be functionally redundant in some of their metabolic processes, and a scientist may be interested in determining when this redundancy occurs or does not occur. 

In addition, the model on the aggregated features may be substantially sparser than the model which only considers the observed features. For instance, in \Cref{fig:treeExample}, suppose that $m(X_i)$ only depends on $X_{i,1}$, $X_{i,4}$, and $X_{i,5}$; furthermore, suppose that $m(X_i)$ actually only depends on their sum, $\tilde X_{i,7}$. Then, the sparsest model in the observed features has 3 covariates whereas allowing for aggregated variables allows for a model with one covariate.
%$\tilde X_7$ as an aggregation of $X_1$, $X_4$, and $X_5$ decreases the number of features from $5$ to $3$. %In many applications, such as in microbiome analysis, the true function may depend only on a small portion of the original features.
Thus, allowing for aggregated variables may have the added benefit of increased statistical efficiency by requiring a smaller number of relevant features. 
This is especially beneficial in the nonparametric setting, where the curse of dimensionality makes each additional relevant feature much more costly.

By construction, $\tilX$ does not have full column rank, since some columns of $\tilX$ are sums of other columns in $\tilX$. Thus, there are many sets $\I$ which may satisfy \Cref{eq:np_regProj}, and the smallest of such sets is not unique. 
%For instance, this may occur when $\I$ includes $v$ as well as $s \in \de(v)$. 
To resolve this ambiguity, we define the \emph{target aggregation set}, denoted as $\Istar$ to be the set of $v \in T$ such that:
\begin{enumerate}[label=(\roman*)]
    \item $\partial_j m \neq 0$ for all $j \in \leaves(v)$,
    \item $\partial_j m = \partial_k m$ for all $j,k \in \leaves(v)$,
    \item $\exists s \in \leaves(\sib(v))$ and $j \in \leaves(v)$ such that $\partial_s m \neq \partial_j m$.
\end{enumerate}
For any $v \in \Istar$, (i) implies that all nodes in $\leaves(v)$ are relevant for predicting $Y$, and (ii) implies that they can be aggregated without losing predictive. If condition (iii) is not satisfied, then the covariates in $\leaves(v)$ should actually be aggregated at a level higher than $v$ itself.
%implies that if a node $v$ is included in the target set, then any other node which is an aggregation of $v$ is not included. 
Although $\I^\star$ cannot be larger than the sparsest set $\I$ which satisfies \Cref{eq:np_regProjLeaves}, it may not be one of the smallest sets that satisfies \Cref{eq:np_regProj}. However, we choose this target set to prioritize interpretability over the statistical efficiency of a potentially smaller but less interpretable set. In particular, if $v \in \Istar$, then $\de(v)\cap \Istar = \emptyset$; thus the practitioner can conclude that $v$ is the sharpest resolution relevant to the prediction task and that data at a more granular level are not needed. %The ``all-or-nothing'' aggregation implies that either all variables in $\leaves(v)$ appear only in the model aggregated as $\tilde X_v$ or else $\tilde X_v$ is not included in the model.

\subsection{Estimator: KR-TEXAS}
\label{sec:estimator}
Our proposed method, Kernel Regression with Tree-EXploring AggregationS estimator (KR-TEXAS), solves the following problem:
\begin{equation}
    \hat{\bm{\gamma}} = \argmin_{\bm\gamma \in \mathbb{R}^{T}}L_n(\g; \tilX, \Y)
 + \lambda_n \sum_{v=1}^T \hat w_v\gamma_v\\
 \label{eq:kr_texas_opt}
\end{equation}
where
 \begin{equation}\small
 L_n(\g; \tilX, \Y) = \frac{1}{n} 
\sum_{i=1}^{n}\left(Y_i - \hat m_{-i}(\tilde X_i; \tilX, \Y, \g) \right)^2, \quad  
     \hat m_{-i}(\tilde X_i; \tilX, \Y, \g)  = \frac{\sum_{j\neq i} \K(\tilde X_j, \tilde X_i) Y_j}{\sum_{j\neq i}\K(\tilde X_j, \tilde X_i)},
 \end{equation}
and $\K(\tilde X_i, \tilde X_j)$ denotes the Gaussian kernel on the aggregated features with inverse bandwidth parameters $\g = \{\g_v\}_{v = 1}^T$ such that $\g_v \geq 0$:
\begin{equation}
\begin{aligned}
    \K(\tilde X_i,\tilde X_j) &= \exp\left(- \sum_{v=1}^T \gamma_v (\tilde X_{i,v} - \tilde X_{j,v})^2 \right).
\end{aligned}
\label{eq:NW_def}
\end{equation}
The estimator can be defined similarly using the Epanechnikov kernel, for which we derive the conditional bias and variance in Supplementary Material \ref{app:BiasVar}. %These expressions show that $\bm \gamma$ determines both the local smoothing (through $\det(\bm A_\gamma^\top\bm A_\gamma)$ in the variance term) and the directions along which the bias depends on the derivatives of $m$ and the density of the predictors, where $\bm A_\gamma = \operatorname{diag}(\sqrt{\bm\gamma})\,\bm A$.
The first term in \Cref{eq:kr_texas_opt} is the Nadaraya-Watson~\citep{nadaraya1964estimating, watson1964smooth} leave-one-out cross-validation objective. The second term is a weighted $L_1$ penalty where $\lambda_n$ is a global regularization parameter and $\hat{w}_{v}$ is a coordinate specific penalization for $\gamma_v$. Similar to the Adaptive Lasso~\citep{zou2006adaptive}, the weights are computed from a pilot estimator before being used for the final estimates. Although the exact value of $\gamma_v$ does not have a straightforward interpretation, a larger value $\gamma_v$ (when fixing the variance of $\tilde X_{i,v}$) implies that predicted values of $Y$ are more sensitive to perturbations in $\tilde X_{i,v}$; furthermore, $\gamma_v = 0$ implies that predicted values of $Y$ do not depend on $\tilde X_{i,v}$ at all. Thus, we let $\hat \I = \{v\,:\, \hat \gamma_v > 0\}$.  
In \Cref{thm:selectionConsistency}, we show that when the weights $\hat w_v$ are chosen in a data-dependent way described below, KR-TEXAS achieves model selection consistency so that $P(\hat \I = \Istar) \rightarrow 1$ as $n\rightarrow\infty$.

\subsubsection{Construction of adaptive weights}
\label{sec:adaptiveweights}
Intuitively, in order to achieve model selection consistency we want weights $\hat w_v$ such that for a sequence of $\lambda_n$, $\lambda_n\hat w_v \rightarrow_p \infty$ for all 
$v \not\in \Istar$ and $\lambda_n\hat w_v \rightarrow_p 0$ for all $v \in \Istar$. To construct such weights, we require pilot estimators of the partial derivatives of $m$ with respect to the observed features. We will use $\partial_v m$ and $\widehat{\partial_v m}$ to denote the derivative of $m$ with respect to $X_{i,v}$ and its estimate, respectively. We then compute the weights as:
\begin{equation}\small
\begin{aligned}
    \label{eq:weights} 
    C_{v, 1} &= \binom{\vert \leaves(v) \vert}{2}^{-1}\sum_{j,k\in \leaves(t)}\left\vert \widehat{\partial_j m} - \widehat{\partial_k m}\right\vert^2_2,\\  
C_{v, 2} &= \frac{1}{\vert \leaves(v)\vert} \sum_{j \in \leaves(v)}\left \vert\widehat{\partial_j m}\right\vert^2_2,\\ 
C_{v, 3} &= \left(\vert \cup_{k\in \sib(v)} \leaves(k) \vert \times \vert \leaves(v)\vert \right)^{-1} \sum_{j\in \leaves(v)} \sum_{k\in \sib(v)} \sum_{l\in \leaves(k)}\left\vert\widehat{\partial_j m} - \widehat{\partial_l m} \right\vert^2_2,\\
\hat w_{v} &= (n^{a_2}C_{v,1})^b + C_{v,2}^{-b} + C_{v,3}^{-b}, 
\end{aligned}
\end{equation}
where $a_2, b >0 $ are the tuning parameters that will be discussed later. 

To provide intuition for the construction of $\hat w_v$, note that each node in $v \in \T$ falls into one of the following 4 mutually exclusive categories where categories 1, 2, and 3 correspond to $\mathcal{T}\setminus {\Istar}$ and category 4 corresponds to $\Istar$:
\begin{enumerate}
    \item 
    For all $j \in \leaves(v)$, $\partial_j m = 0$,
    \item 
    There exists $j \in \leaves(v)$ such that $\partial_j m \neq 0$, but there exists $j,k \in \leaves(v)$ such that $\partial_j m \neq \partial_k m$, 
    \item 
    For all $j,k \in \leaves(v)$, $\partial_j m = \partial_k m$ and $\partial_j m \neq 0$, but $\partial_j m = \partial_k s$ for all $s \in \bigcup_{k \in \sib(v)}\leaves(k)$,
    \item 
    For all $j,k \in \leaves(v)$, $\partial_j m = \partial_k m$ and $\partial_j m \neq 0$, but $\exists s \in \bigcup_{k \in \sib(v)}\leaves(k)$ such that
    $\partial_j m \neq \partial_k s$. 
\end{enumerate}
Suppose that for some $a_1 > 0$, $\max_{v \in [p]}\vert \widehat{\partial_v m} - \partial_v m \vert_2 \lesssim n^{-a_1}$. If $v$ is in Category 1, it is not in $\I^\star$ since none of its leaves effect $Y$; in this case, $C_{v,2} = O(n^{-a_1})$ so $w_{v,n} = \Omega(n^{a_1b})$. If $v$ is in Category 2, it is not in $\I^\star$ because it has leaves which do not share the same derivative and thus those leaves should not be aggregated; in this case, $C_{v,1} = \Omega(1)$ so $w_{t,n} = \Omega(n^{a_2b})$. If $v$ is in Category 3, even though all nodes in $\leaves(v)$ are active and have the same derivative, the leaves of $v$'s siblings have the same derivatives as $\leaves(v)$. Thus the leaves of $v$ should actually be aggregated at a level higher than $v$ (i.e., at some $\an(v)$) and $v$ itself should not be in $\I^\star$. In this case, $C_{v,3} = O(n^{-a_1})$ so $w_{t,n} = \Omega(n^{a_1b})$. Finally, if $v$ is in Category 4, it should be in $\I^\star$. In this case, $C_{v,2}$ and $C_{v,3}$ are both $\Omega(1)$ while $C_{v,1} = O(n^{-a_1b})$; thus, $w_{t,n} = O(n^{(a_2-a_1)b})$. If $0< a_2 < a_1$, then the weights for nodes in Categories 1-3 will increase with $n$ and the weights for nodes in Category 4 will decrease with $n$.

\subsubsection{Practical concerns}\label{sec:practicalConcern}

Estimating derivatives is known to be more difficult than estimating the nonparametric regression function itself; see e.g., \cite{yatracos1989estimation}. Nonetheless, our method can substantially improve our predictions of $Y$ even with noisy pilot estimators; this is similar to ``plug-in'' methods which use estimates of derivatives to select a bandwidth~\cite{gasser1991flexible, ruppert1995effective}. Many approaches for estimating derivatives have been proposed~\citep{stone1982optimal, ruppert1994multivariate, fan1995data, de2013derivative, dai2016optimal, wang2019derivative, liu2023estimation, liu2026optimal}, and Theorem~\ref{thm:selectionConsistency} allows for any method which estimates the derivatives sufficiently well. 

%. Other proposed procedures directly estimate the derivative but primarily focus on the univariate setting~\citep{}.

In practice, we find that the following procedure for calculating weights performs well empirically by adapting to sparsity in $m$. We first learn a metric, $\check \g$, over just the observed features (i.e., leaf nodes) using leave-one-out cross validation: 
\begin{equation}  \label{eq:leaf_opt}
\check{\g} = \argmin_{\bm\gamma \in \mathbb{R}^{p}}  \frac{1}{n} 
\sum_{i=1}^{n}\left(Y_i - \hat m_{-i}(\X, \Y; \gamma) \right)^2.
 \end{equation}
Next, we use the learned metric to estimate the derivative at each point $X_i$ using local linear regression; using $\check \g$ allows for better estimates because it may adapt to the potential sparsity in $m$~\citep{hall2007nonparametric, conn2019oracle}.
Furthermore, when estimating a derivative, it is generally beneficial to ``oversmooth'' relative to what might be optimal for estimating the regression function itself \citep{fan1996local}. Thus, instead of directly using $\check \g$, we use $\check \g^z = \{ \check \gamma_v^{z}\}_{v = 1}^p$ for some $0 < z < 1$. The optimal amount of oversmoothing depends on the true sparsity of the regression function, which is unknown; however, in Section~\ref{sec:numerics}, we see that setting $z = 3/4$ works well empirically. Putting everything together, we calculate $\hat \beta_{i,v}$ as an estimate of $\left.\partial_v m\right\vert_{X_i}$ by solving:     
\begin{equation}  \label{eq:leaf_opt_llr}
\hat \beta_i = \arg\min_{\beta\in \mathbb{R}^{p+1}}   
\sum_{j=1}^n\left(Y_j - \beta_0 + \sum_{v = 1}^p \beta_v (X_{i,v} - X_{j,v})\right)^2K_{\check \g^z}(X_i, X_j).
\end{equation}

Furthermore, when plugging in the estimated derivatives to calculate the weights, the estimates for points near the boundary will typically be less reliable than points in the ``interior.'' Thus, for each observation $X_i$, we calculate $K_i = \sum_{j \neq i} K_{\check \g}(X_i, X_j)$, and let $\mathcal{J}$ denote the set of $m = n/10$ points with the largest $K_i$ scores. We subsequently calculate $C_{v,1},  C_{v,2}, C_{v,3}$ using only the points in $\mathcal{J}$; e.g., $  C_{v,1} = \binom{\vert \leaves(v) \vert}{2}^{-1}\sum_{u, w \in \leaves(v)}\sum_{i \in \mathcal{J}} \vert \hat \beta_{i,w} - \hat \beta_{i,u} \vert$. 
For model selection consistency, the exact values of $C_{v,1}, C_{v,2}, C_{v,3}$ are not crucial as long as they converge to either $0$ or some non-zero value correctly. If $\partial_v m$ and $\partial_u m$ only differ in regions near the boundary, the modified weights may indeed differ qualitatively from the weights calculated using the entire domain. However, if $\partial_u m \neq \partial_v m$ somewhere implies that $\partial_u m \neq \partial_v m$ almost everywhere, then the modified weights will suffice. This holds, for example, if $m$ is a real analytic function in an open connected domain. %Ultimately, this choice can be viewed as taking fewer, but potentially higher quality, measurements of the relevant derivatives.  

%It is important to note that this optimization problem is non-convex, and the solution to the optimization is not guaranteed to be a global minimizer of the loss function. Our optimization algorithm for a fixed lambda is described in Algorithm \ref{alg:gaussianKRTEXAS_fixedlambda}, and two-stage training procedure is given in Algorithm \ref{alg:krtexas_fit_full}.

Both Eq.~\eqref{eq:kr_texas_opt} and Eq.~\eqref{eq:leaf_opt} are non-convex. However, we may still find local minima using gradient based methods and use random initializations to hopefully find a global minimum. We can calculate the gradient analytically, and our implementation in \texttt{R} uses L-BFGS with box constraints~\citep{byrd1995limited}. The simulations in \Cref{sec:numerics} show that 10-30 random restarts typically suffice for good empirical performance. In particular, \Cref{thm:selectionConsistency} shows that asymptotically any stationary points of~\Cref{eq:kr_texas_opt} will not contain false positives; however, we require a global minimum to ensure that there are no false negatives. 

To select the global regularization parameter $\lambda_n$ in \Cref{eq:kr_texas_opt}, we estimate the weights $\hat w_{v}$ once using all the data. Then, we use $K=5$ fold cross validation to select $\lambda_n$. The estimated $\hat \g$ and $\hat \I$ is calculated by solving \Cref{eq:kr_texas_opt} using the selected $\lambda_n$ and all observed data. We note that $\hat \I$ may not satisfy the desired constraint that $v \in \hat \I$ implies that $\de(v) \cap \hat \I = \emptyset$. In the simulations, we leave $\hat \I$ unchanged, but in practice one could post-process the set to explicitly enforce the constraint. 

\begin{comment}
into a set which preserves this property in two different ways; for any  $v \in \hat \I$ such that $\de(v) \cap \hat \I \neq \emptyset$ we may:
\begin{enumerate}
    \item \emph{Aggregate up}: Remove any $\de(v)$ from $\hat \I$
    \item \emph{De-aggregate down}: Remove $v$ from $\hat \I$ and add the set 
    \[\hat D_v = \textcolor{red}{\arg}\min_{D \subseteq \de(v)} \vert D \vert \qquad \text{ s.t. } \leaves(D) = \leaves(v) \text{ and } \de(D) \cap \hat \I = \emptyset. \]
\end{enumerate}
\end{comment}

\begin{algorithm}[t]
\caption{KR-TEXAS Two-Step Procedure}
\begin{algorithmic}[1]
\Statex \textbf{Step A: Adaptive weight construction}
\For{$r \in [M_1]$}
\State Use random initialization to get a local optimum  $\hat \g_{r}$ of~\Cref{eq:leaf_opt}
\EndFor
\State Set $\hat \g_{0} =  \argmin_{r}  \frac{1}{n} 
\sum_{i=1}^{n}\left(Y_i - \hat m_{-i}(\X, \Y; \hat\g_r) \right)^2.$
\State Estimate gradients $\{\widehat{\partial_j m}\}_{j=1}^p$ using LLR via~\Cref{eq:leaf_opt_llr} with $\hat \g_{0}$
\State Compute weights $\hat w_v$ via  \Cref{eq:weights} for all $v \in \mathcal{T}$ using $\{\widehat{\partial_j m}\}_{j=1}^p$
\Statex \textbf{Step B: Penalized fitting with adaptive weights}      
\State Split $[n]$ into $k = 1\ldots K$ folds with train/test sets denoted as $D_{train}^{(k)}$ and $D_{test}^{(k)}$ 
\For{$m \in [M_2]$}
\For{$\lambda \in \Lambda$ and $k = 1:K$}
        \State Use random initialization then warm starts to get a local optimum $\hat \g_{\lambda,m}^{(k)}$ of
        \[
        %\mathcal{L}(\g; \tilX_{D_{train,m}^{(k)}}, \Y_{D_{train,m}^{(k)}},\lambda)= 
        \frac{1}{\vert D_{train}^{(k)}\vert}\sum_{i\in D_{train}^{(k)}} (Y_i - \hat m_{-i}(\tilde X_i; \tilX_{D_{train}^{(k)}}, \Y_{D_{train}^{(k)}}, \g))^2
        + \lambda \sum_{v=1}^T \hat w_v \g_v.
        \]
\EndFor
\EndFor
\State Select $(\lambda^\star, m^\star) = \arg \min_{\lambda \in \Lambda, m} \sum_{k=1}^K \sum_{i \in D_{test,m}^{(k)}}(Y_i - \hat m_{-i}(\tilde X_i; \tilX_{D_{train}^{(k)}}, \Y_{D_{train}^{(k)}}, \hat\g_{\lambda, m}^{(k)}))^2$

\State Use $m^\star$ random initialization to get local optimum $\hat \g$ of
\[
        \mathcal{L}(\g; \tilX, \Y,\lambda^\star)
        = \frac{1}{n}\sum_{i\in [n]} (Y_i - \hat m_{-i}(\tilde X_i; \tilX, \Y, \g))^2
        + \lambda^\star \sum_{v=1}^T \hat w_v \g_v.
\]
\end{algorithmic}
\label{alg:krtexas_fit_overall}
\end{algorithm}

The pseudocode for the entire procedure is given in \Cref{alg:krtexas_fit_overall}, and additional implementation details are given in Supplementary Material \ref{app:implementation}. For each fold $k$, let 
$D_{\text{test}}^{(k)} \subset [n]$ denote the indices assigned to 
the $k$-th held-out fold, and let 
$D_{\text{train}}^{(k)} = [n] \setminus D_{\text{test}}^{(k)}$ 
denote the corresponding training indices.
A user-friendly software implementation of KR-TEXAS can be found at \url{https://github.com/sithijamanage/krtexas}.

\section{Theoretical guarantees}\label{sec:theory}
We now show that KR-TEXAS can recover the target aggregation set with probability going to $1$ as $n \rightarrow \infty$. Specifically, we show that asymptotically, there will be no false positives for any stationary point of \Cref{eq:kr_texas_opt}, and when $\hat \I$ is estimated from the the global minimum, then $P(\hat \I = \I^\star) \rightarrow 1$. All proofs are given in Supplementary Material \ref{app:selectionConsistency}. 

The first condition in \Cref{assump:bounded} is a common requirement that the covariates and responses are bounded, and the second condition essentially requires that the covariates have a density uniformly bounded from below. \Cref{assump:contFunc} requires that the true regression function $m$ be sufficiently smooth. 
\begin{assumption}
\label{assump:bounded}
Suppose $\vert X_{i} \vert_2 < B$ and $\vert Y_i \vert < B$. Furthermore, the random vector $X_i$ has support on a set $\mathcal{X}$ where $\min_{x \in \mathcal{X}}\E_{X_i}(\K(\tilde x,\tilde X_i)) \geq c_{f}\prod_{v: \g_v \neq 0}\g_v^{-1/2}$ for some constant $c_{f}>0$ for all $\g$. 
\end{assumption}

\begin{assumption}
 \label{assump:contFunc} The function $m$ belongs to a H\"{o}lder class with smoothness $\beta > 1$.
\end{assumption}

In \Cref{sec:practicalConcern} we detail a specific procedure for estimating the derivatives; however, any method can be used as long as the other tuning parameters are chosen to satisfy Assumption~\ref{assump:tuning}. For example, if $\beta = 2$ under \Cref{assump:contFunc}, when estimating $m$ for the observed features $\E(\Vert \hat m - m \Vert_2) \lesssim n^{-4/(4+p)}$ and using the derivatives of $\hat m$ then it produces a rate of $\E(\Vert \widehat{\partial_j m} -\partial_j m \Vert_2) \lesssim n^{-2/(4+p)}$ \citep{yatracos1989estimation}. In general, setting $b = 1$ with $0 <d < a_2 <a_1$ will suffice. However, $b$ could also depend on $n$ (e.g., $b = \log(n)$) which would allow for more flexibility in choosing $d$. Although the theory assumes that $\lambda_n \asymp n^{-d}$ for some known $d$, in practice, we select it through cross-validation. 

\begin{assumption}\label{assump:tuning}
Suppose the pilot estimators satisfy $\max_{j \in [p]}\left\Vert \widehat{\partial_j m} -  \partial_j m  \right\Vert_2^2\lesssim n^{-a_1}$ for $a_1 > 0$ and $\lambda_n \asymp n^{-d}$ for $d>0$. As $n\rightarrow \infty$, the tuning parameters $a_2,$ $b,$ and $d$ satisfy
    \begin{align*}
        a_2b-d>0, \quad a_1b - d >0, \quad \text{ and }  \quad \left(a_2 - a_1\right)b - d < 0.
\end{align*}
\end{assumption}

Under these assumptions, \Cref{lem:noFalsePos} shows that, asymptotically, any stationary point of KR-TEXAS does not produce false positives. 
\begin{restatable}{lemma}{noFalsePos}
\label{lem:noFalsePos}
Let $G^{\dagger}$ denote the set of all stationary points of~\Cref{eq:kr_texas_opt} and suppose Assumptions~\ref{assump:bounded}, \ref{assump:contFunc}, and \ref{assump:tuning} hold. Then, as $n$ increases
\[P\left(\max_{\g \in G^\dagger}\max_{k \not \in \I^\star}\g_k = 0\right) \rightarrow 1.\] 
\end{restatable}

For a set $\I \subset [T]$, we say that $\I$ is a super-model of $\I^\star$ if the row span of $A_\I$ contains the row span of $A_{\I^\star}$; this implies that $\text{span}(\tilX_{\cdot, \I^\star}) \subseteq \text{span}(\tilX_{\cdot, \I})$. Similarly, we say that $\I$ is a strict non-super-model of $\I^\star$ if the row span of $A_\I$ does not contain the row span of $A_{\I^\star}$; this implies that $\text{span}(\tilX_{\cdot, \I^\star}) \not \subseteq \text{span}(\tilX_{\cdot, \I})$. Assumption~\ref{assump:seperation} is a ``minimum signal strength'' assumption which ensures that the optimal predictions for a strict non-super-model must be strictly worse than the optimal predictions for the targeted aggregation set. 
\begin{assumption}\label{assump:seperation}
For all $\I$ which are strict non-super-models of $I^\star$, for some $\delta > 0$:
\[\E\left(\left[Y_i-\E(Y_i \mid P_{\I}(\tilde X_i)\right]^2\right) > \E\left(\left[Y_i-\E(Y_i \mid P_{\I^\star}(\tilde  X_i)\right]^2\right) + \delta. \] 
\end{assumption}

Finally, we define a restricted parameter space in order to derive a uniform law of large numbers that ensures that there are no false positives in \Cref{thm:selectionConsistency}. For some $C_g > 0$, let
\begin{equation}
    \begin{aligned}
        G_{n}  &= \left\{\bm\gamma \, :\, \vert \bm\gamma\vert_\infty \leq C_g n^{2/(2\beta+\vert \gamma \vert_0)}\right\}\\
        \check G_{n}  &= \left\{\g \in G_n : \text{ such that } \I = \{v \,:\,\gamma_v > 0\} \text{ is a strict non-super-model of } \I^\star \right\}.
    \end{aligned}
\end{equation}
Because $\g_v$ is equivalent to the squared bandwidth in the typical parameterization, $G_n$ contains the optimal (population level) bandwidths for any $\I\subset \mathcal{T}$ which are of the order $n^{2/(2\beta+\vert \I \vert)}$. In practice, we do not explicitly enforce this constraint, but instead implicitly enforce this by using gradient-based methods initialized with small starting values. This seems to work well empirically.

\begin{restatable}{theorem}{modelSelect}
\label{thm:selectionConsistency}
Suppose Assumptions~\ref{assump:bounded}, \ref{assump:contFunc}, \ref{assump:tuning}, and \ref{assump:seperation} hold. Let $\hat \I = \{v : \hat \g_v > 0\}$ where $\hat \g = \arg \min_{\g \in G_{n}} L_n(\g; \tilX, \Y) + \lambda_n\sum_{v}\hat w_v\g_v$. Then,
\begin{equation}
P\left(\hat \I = \Istar \right) \rightarrow  1.
\end{equation}
\end{restatable}

%Directly applying \Cref{lem:unifControl} does not imply that the selected $\hat \g$ achieves the oracle rate of $n^{2\beta(2\beta + \vert I^\star \vert}$. However, in contrast to previous results which show this..., \Cref{}  

Since \Cref{thm:selectionConsistency} implies consistency in model selection, an estimator with the the oracle rate could be formed by splitting the sample and using the first half of the data to estimate $\hat \I$. Conditional on $\hat \I = \Istar$, which occurs with probability $1-o(1)$, estimating $\hat m_{\hat \I}$ using the second half of the data would achieve the oracle rate of $n^{-2\beta/(2\beta+\vert \Istar \vert)}$. %Of course, sample splitting comes at a cost in finite samples. 
However, sample splitting may come at a large cost in finite samples, and in \Cref{sec:numerics}, we see that our method (without sample splitting) is still empirically comparable to oracle methods.  

% The rate in~\Cref{cor:optRate} is optimal up to log factors for nonparametric regression with $\vert \Istar \vert$ covariates when $m$ is only assumed to be Lipschitz~\citep{stone1982optimal}. 
% However, we have assumed higher levels of smoothness in order to estimate the derivatives. 

\section{Numerical experiments}\label{sec:numerics}
\subsection{Simulation Design}

%To assess the performance of our estimator, we first specify a simulation framework that includes the construction of a tree, the selection of active nodes, the generation of covariates $\bm X$, and the mechanism by which the outcome $Y$ is derived from $\bm X.$

For simulations, we first generate covariates with $p = 128$ using a scaled Gaussian copula. Specifically, we draw $Z_i \sim \mathcal{N}\left(0, \Sigma \right)$, then transform the data via $X_i = 2\Phi( Z_i)-1$, where $\Phi(\cdot)$ is the standard normal CDF applied element-wise. We set $\Sigma$ to be either (1) the identity, (2) a Toeplitz matrix with $\Sigma_{ij} = 0.4^{|i-j|}$, or (3) a tridiagonal matrix with $\Sigma_{ij} = 1_{\{i = j\}} + .4 \times 1_{\{i \neq j\}}$. The covariates are structured according to a full binary tree $\mathcal{T}$ and so $T = 2p-1 = 255$.
$Y_i$ is constructed from five groups of covariates, each corresponding to a different node at varying depths in $\mathcal{T}$. Let $S_{i,1}= \sum_{j=1}^{4} \tilde X_{ij},\,
S_{i2} = \sum_{j = 33}^{34} \tilde X_{ij},\,
S_{i3} = \sum_{j = 65}^{69} \tilde X_{ij},\,
S_{i4} = \tilde X_{i,97},\,
S_{i5}= \tilde X_{i,98}$.
We consider three different settings for generating $Y_i$:
\begin{itemize}
    \item Nonlinear 1\,: $Y_i
= S_{i1}^2
+ 5\cos(S_{i2})
- S_{i4}
+ 10\left(1 + S_{i3}^2\right)^{-1}\, S_{i5}^2 + \varepsilon_i$
 \item Nonlinear 2\,: $Y_i
= S_{i1}^3
+ 5\sin(S_{i2})
- 2S_{i3}^2
- S_{i4}
+ 0.5S_{i5}^5
+ \varepsilon_i$
    \item Linear\,\,\,\,\,\,\,\,\,\,\,\,\,\,\,: $Y_i = 2S_{i1} +  5S_{i2} + S_{i3} - 3S_{i4} - 2S_{i5} + \varepsilon_i$
\end{itemize}
where $\varepsilon_i\sim N(0,(0.01\times \sigma_Y)^2)$, and $\sigma_Y$ denotes the standard deviation of the noiseless responses. We consider the sample sizes of $n =500, 1000, 2500, 5000$ and use 200 replicates for each setting of parameters and sample size. These sample sizes reflect those of large microbiome cohort studies such as the Human Microbiome Project \citep{sa2012framework}, the American Gut Project \citep{mcdonald2018american}, and LifeLines-DEEP \citep{zhernakova2016population}.

We compare the performance of KR-TEXAS to an ``oracle'' version in which Nadaraya-Watson with a cross-validated anisotropic bandwidth is applied to the ``true'' aggregated features $\{S_{i,v}\}_{v=1}^5$ (KR-TEXAS Oracle), Nadaraya-Watson on the ``true'' aggregated features $\{S_{i,v}\}_{v=1}^5$ with a fixed isotropic MSE optimal bandwidth of $n^{-1/(4 + \vert \Istar \vert)}$ (NW Oracle), Nadaraya-Watson on the original features $\X$ with an anisotropic bandwidth selected by cross-validation (NW+ML), Rare feature aggregation (RARE)~\citep{bien2021rare}, the lasso on $\tilX$ (LASSO AX)~\citep{tibshirani1996regression}, the lasso on $\X$ (LASSO), Nadaraya-Watson on $\tilX$ with an isotropic bandwidth selected by cross-validation  (NW AX), and Nadaraya-Watson on $\X$ with an isotropic bandwidth selected by cross-validation (NW).

To evaluate the prediction performance, we calculate the root mean squared prediction error (RMSE) \[\sqrt{(1/n_\text{test})\sum_{i=1}^{n_\text{test}}(\E(Y\mid X_i) - \hat m(X_i))^2}\] on a fixed test set of size $n_{\text{test}} = 1,000$. For methods that select a specific model (KR-TEXAS, RARE, LASSO AX), we also evaluate variable selection using the sensitivity (SN), specificity (SP), precision (Prec), and negative predictive value (NPV). Additional details for the simulation settings are in Supplementary Material \ref{app:NumericalExpDetails}.

\subsection{Results}

\subsubsection{Prediction Performance}
\Cref{fig:rmse_main} presents the prediction performance of KR-TEXAS Oracle, KR-TEXAS, NW Oracle, NW + ML, RARE, LASSO AX, and NW across the three regression function settings for identity and tridiagonal covariance. \Cref{fig:rmse_all} in Supplementary Material \ref{app:NumericalExpDetails} contains the prediction performance results for all methods and covariance settings. We show a box-plot for the 200 replicates and use a white diamond to indicate the median.
%and each point represents a replicate with RMSE beyond 1.5 units times the interquartile range.
The standard deviation of $Y$ for each setting is marked with a dashed line; this would correspond to the loss of simply predicting $\bar Y$ for each point.

\begin{figure}[t]
    \centering
\includegraphics[width=\linewidth]{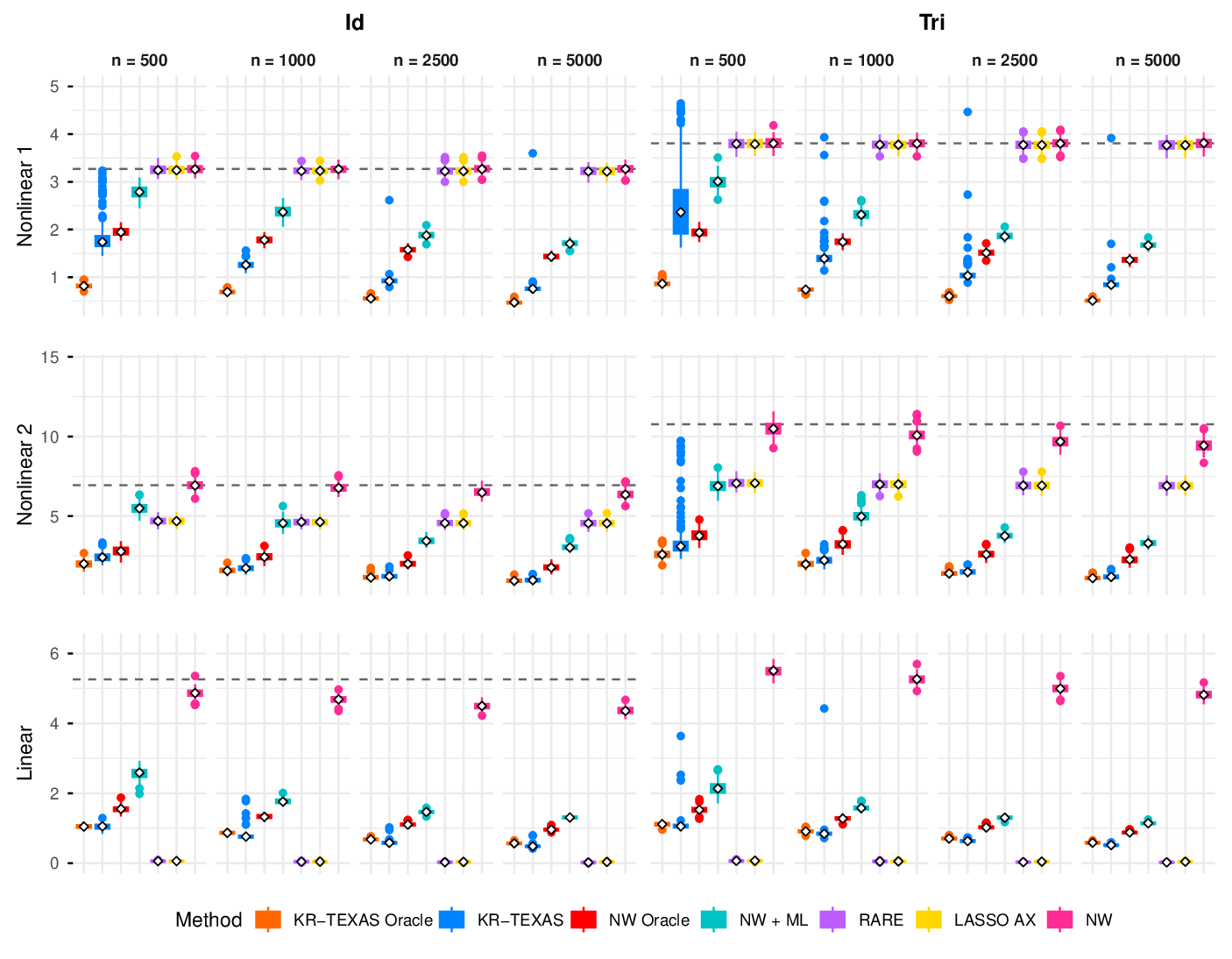}
    \caption{Prediction performance (RMSE) of all methods with identity and tridiagonal covariance.}
    \label{fig:rmse_main}
\end{figure}

In both nonlinear settings, across all sample sizes and covariance frameworks, KR-TEXAS and KR-TEXAS Oracle consistently achieve the lowest median RMSE.
 The Nadaraya-Watson estimators with oracle (NW Oracle) and anisotropic bandwidths (NW + ML) perform competitively but have larger predictive errors than the two KR-TEXAS methods; this shows the benefit of combining feature aggregation and metric learning. In the first nonlinear setting, where the functions are mostly even functions and contain an interaction term, the linear methods (RARE and LASSO) and isotropic bandwidth Nadaraya-Watson on $\bm X$ (NW) do not do much better than simply predicting $\bar Y$. However, in the second nonlinear setting, which contains additive, mostly odd functions, the linear methods are able to capture the signal and improve over predicting the mean. In settings with sample sizes of $500$ and $1000$, KR-TEXAS exhibits higher variability than the other methods, perhaps due to the non-convex optimization landscape. %As sample size increases, we observe that the predictive gap between the two groups \{KR-TEXAS Oracle, KR-TEXAS, NW Oracle, NW + ML\} and \{RARE, LASSO AX, NW\} widens.
Unsurprisingly, when the linear model is correct, the linear methods (RARE and LASSO-AX) outperform the nonparametric approaches in terms of prediction error. Notably, RARE achieves the lowest RMSE values across every sample size and covariance setting. However, in this setting, KR-TEXAS still outperforms the nonparametric methods which do not aggregate the covariates.

\subsubsection{Variable Selection Performance}

\begin{table}[t]\footnotesize
\caption{Variable selection performance across different settings}
\centering\label{tab:nonlinearTab}
\begin{tabular}{ccc|@{\,}r@{\,\,}r@{\,\,}r@{\,\,}r@{\,}|@{\,}r@{\,\,}r@{\,\,}r@{\,\,}r@{\,}|@{\,}r@{\,\,}r@{\,\,}r@{\,\,}r@{\,}|@{\,}r@{\,\,}r@{\,\,}r@{\,\,}r@{\,}}
& & & \multicolumn{4}{c}{$n=500$} & \multicolumn{4}{c}{$n=1000$} & \multicolumn{4}{c}{$n=2500$} & \multicolumn{4}{c}{$n=5000$} \\
& & & SN & SP & Prec & NPV & SN & SP & Prec & NPV & SN & SP & Prec & NPV & SN & SP & Prec & NPV \\
\midrule
 \multirow{6}{*}{\rotatebox{90}{Nonlinear 1}}& \multirow{3}{*}{Id} & KRT & \textbf{1.00} & .35 & .03 & \textbf{1.00} & \textbf{1.00} & .61 & .05 & \textbf{1.00} & \textbf{1.00} & .72 & .07 & \textbf{1.00} & \textbf{1.00} & .87 & .13 & \textbf{1.00} \\
\cline{3-19}
& & LASSO & .19 & \textbf{.98} & \textbf{.15} & .98 & .22 & \textbf{.98} & \textbf{.20} & .98 & .23 & \textbf{.98} & \textbf{.18} & .98 & .23 & \textbf{.98} & \textbf{.17} & .98 \\
\cline{3-19}
& & RARE & .17 & .97 & .11 & .98 & .22 & .97 & .14 & .98 & .22 & .97 & .14 & .98 & .22 & .97 & .12 & .98 \\
\cline{2-19}
& \multirow{3}{*}{Tri} & KRT & \textbf{.85} & .56 & .04 & \textbf{.99} & \textbf{1.00} & .63 & .05 & \textbf{1.00} & \textbf{1.00} & .66 & .06 & \textbf{1.00} & \textbf{1.00} & .93 & \textbf{.22} & \textbf{1.00} \\
\cline{3-19}
& & LASSO & .15 & \textbf{.98} & \textbf{.14} & .98 & .21 & \textbf{.98} & \textbf{.18} & .98 & .23 & \textbf{.98} & \textbf{.17} & .98 & .22 & \textbf{.98} & .17 & .98 \\
\cline{3-19}
& & RARE & .14 & .97 & .09 & .98 & .20 & \textbf{.98} & .15 & .98 & .22 & .97 & .13 & .98 & .22 & \textbf{.98} & .17 & .98 \\
\midrule
\multirow{6}{*}{\rotatebox{90}{Nonlinear 2}}
&\multirow{3}{*}{Id} & KRT & \textbf{.95} & .40 & .03 & \textbf{1.00} & \textbf{.92} & .53 & .04 & \textbf{1.00} & \textbf{.96} & .54 & .04 & \textbf{1.00} & \textbf{.98} & .55 & .04 & \textbf{1.00} \\
\cline{3-19}
& & LASSO & .59 & \textbf{.96} & \textbf{.23} & .99 & .65 & \textbf{.96} & \textbf{.23} & .99 & .69 & \textbf{.96} & \textbf{.25} & .99 & .74 & \textbf{.96} & \textbf{.26} & .99 \\
\cline{3-19}
& & RARE & .56 & \textbf{.96} & .21 & .99 & .65 & .94 & .18 & .99 & .69 & .94 & .19 & .99 & .75 & .94 & .21 & .99 \\
\cline{2-19}
& \multirow{3}{*}{Tri} & KRT & \textbf{.87} & .50 & .03 & \textbf{.99} & \textbf{.88} & .57 & .04 & \textbf{1.00} & \textbf{.86} & .63 & .04 & \textbf{1.00} & \textbf{.86} & .64 & .05 & \textbf{1.00} \\
\cline{3-19}
& & LASSO & .50 & \textbf{.96} & \textbf{.22} & \textbf{.99} & .55 & \textbf{.96} & \textbf{.22} & .99 & .63 & \textbf{.95} & \textbf{.20} & .99 & .66 & \textbf{.95} & \textbf{.19} & .99 \\
\cline{3-19}
& & RARE & .48 & \textbf{.96} & .18 & \textbf{.99} & .51 & .95 & .18 & .99 & .60 & .94 & .17 & .99 & .65 & .91 & .12 & .99 \\
\midrule
\multirow{6}{*}{\rotatebox{90}{Linear}} &
\multirow{3}{*}{Id} & KRT & \textbf{.81} & .60 & .04 & .99 & \textbf{.80} & .85 & .10 & \textbf{1.00} & \textbf{.80} & .70 & .05 & .99 & \textbf{.80} & .64 & .04 & .99 \\
\cline{3-19}
& & LASSO & .80 & \textbf{.94} & \textbf{.21} & \textbf{1.00} & \textbf{.80} & \textbf{.94} & \textbf{.20} & \textbf{1.00} & \textbf{.80} & \textbf{.99} & \textbf{.56} & \textbf{1.00} & \textbf{.80} & \textbf{.99} & \textbf{.75} & \textbf{1.00} \\
\cline{3-19}
& & RARE & .80 & .91 & .15 & \textbf{1.00} & \textbf{.80} & \textbf{.94} & \textbf{.20} & \textbf{1.00} & \textbf{.80} & .91 & .15 & \textbf{1.00} & \textbf{.80} & .94 & .22 & \textbf{1.00} \\
\cline{2-19}
& \multirow{3}{*}{Tri} & KRT & \textbf{.83} & .71 & .05 & \textbf{1.00} & \textbf{.81} & .81 & .08 & \textbf{1.00} & \textbf{.82} & .58 & .04 & .99 & \textbf{.81} & .53 & .03 & .99 \\
\cline{3-19}
& & LASSO & .80 & \textbf{.94} & \textbf{.21} & \textbf{1.00} & .80 & \textbf{.94} & \textbf{.20} & \textbf{1.00} & .80 & \textbf{.99} & \textbf{.55} & \textbf{1.00} & .80 & \textbf{.99} & \textbf{.71} & \textbf{1.00} \\
\cline{3-19}
& & RARE & .80 & .92 & .17 & \textbf{1.00} & .80 & \textbf{.94} & \textbf{.20} & \textbf{1.00} & .80 & .93 & .19 & \textbf{1.00} & .80 & .94 & .20 & \textbf{1.00} \\
\hline
\end{tabular}
\end{table} 

The variable selection results in \Cref{tab:nonlinearTab} show that our method generally outperforms the linear methods in both nonlinear settings, even when the sample size is small. 
Notably, KR-TEXAS achieves the highest sensitivity and negative predictive value across all values of $n$ in both covariance settings. Although the linear methods have better precision and specificity at smaller sample sizes, KR-TEXAS achieves comparable performance as the sample size increases (without sacrificing sensitivity or negative predictive value). In contrast, both linear methods have consistently low sensitivity, especially in the Nonlinear 1 case. %while maintaining exceptionally high specificity. %Overall, in the nonlinear settings, KR-TEXAS has strong variable selection performance when compared to RARE and the lasso.
When the true model is linear, RARE and LASSO-AX generally outperform KR-TEXAS in support recovery as expected. LASSO-AX and RARE have exceptionally high specificity, especially when $n \geq 2500$, and unlike the first nonlinear case, it does not come at the cost of sensitivity. In particular, while, KR-TEXAS achieves (or ties with) the highest sensitivity across all sample sizes, LASSO-AX achieves the best specificity, precision, and negative predictive value (with RARE only slightly lacking in precision). 

\Cref{tab:allCasesTab} in Supplementary Material \ref{app:NumericalExpDetails} contains the full variable selection performance results for all covariance settings.

\section{Data Analysis}\label{sec:analysis}
We now apply our procedure to a microbiome dataset. Although our proposed procedure does not specifically address the compositional nature of the microbiome covariates, we see that interesting structure is recovered nevertheless.

Acetic, propionic and butyric acid are short chain fatty acids (SCFAs) derived from intestinal microbial fermentation of indigestible foods in the human gut.
There is growing evidence that SCFAs are crucial for disease development and health maintenance \citep{sellin1998short, tan2014role}. 
%and are known to affect intestinal motility, barrier function, and host metabolism \citep{canfora2015short, zhang2023short}. 
We utilize data from the \texttt{curatedMetagenomicData} \texttt{R}  package \citep{cmd2017} repository to obtain species-level relative abundance profiles and HUMAnN3 \citep{humann3} pathway abundance data from healthy control stool samples. HUMAnN3 is a functional profiling pipeline that maps metagenomic reads to gene families and reconstructs sample-specific metabolic pathway abundances. For participants with repeated observations, only the earliest visit was retained to avoid longitudinal dependence. We include any study with $>1,000$ qualifying samples to ensure adequate sample size and consistency across cohorts; this results in the AsnicarF 2021 and LifeLinesDeep 2016 cohorts \citep{asnicar2021microbiome, zhernakova2016population}.

We construct short-chain fatty acid (SCFA) scores directly from functional (HUMAnN3) pathway abundance profiles provided by the \texttt{curatedMetagenomicData} \texttt{R} package. For each sample, we extracted pathways whose names contained SCFA-related substrings (“butyr”, “butanoate”, “propionat”, “propanoate”, or “acetat”), thereby capturing butyrate-, propionate-, and acetate-producing pathways. 
The abundances of these selected pathways are then summed for each sample, after adding a pseudocount of $10^{-6}$ to avoid zeros, to obtain an aggregate SCFA pathway abundance. This aggregate is log-transformed and finally standardized across samples. The resulting standardized log-sum defines the SCFA score used as the response in our analysis.

For the observed covariates, we select the top $50$ most prevalent species by relative abundance. To construct the taxonomic tree, we take the union of the full taxonomic lineage (kingdom through species) for each of the $50$ species. If $v$ is an internal node with one child, we remove it and add an edge from the parent of $v$ to the child of $v$. This results in a tree with $85$ total nodes that include the observed covariates.

 %Using this set, we build the tree matrix $\bm A$ as defined in Section \ref{sec:methods}. %To the best of our knowledge, although these cohorts are widely used benchmark datasets in human gut metagenomics, our study is the first to apply an adaptive tree‑aggregation framework to model these data. 
When applied to the data, our procedure selects 40 (possibly aggregated) covariates. As a rough measure of feature importance, we use $\hat \gamma_v \times  \var(\tilde X_{i,v})$ and list the top 20 features in~\Cref{tab:scfa}. A table containing all 41 selected covariates is provided in Supplementary Material \ref{app:dataAnalysis}.
Many taxa identified by KR-TEXAS align with well‑established findings in the microbiome literature in relation to SCFAs; others are comparatively understudied, perhaps suggesting promising avenues for new biological discovery. Additional details are provided in Supplementary Material \ref{app:dataAnalysis}.

The gut SCFA production ecosystem is a tightly coupled metabolic network organized as a stepwise pipeline \citep{flint2012role}. Our estimator assigns the strong predictive value of SCFAs to the abundances of the phyla Bacteroidetes and Firmicutes, which have been shown to contribute to colonic SCFA production in human fermentation and microbiome studies \citep{riviere2016bifidobacteria,fu2019nondigestible}.  Complex dietary fibers are first degraded by taxa within Bacteroidetes, which break down simple sugars and oligosaccharides \citep{flint2012role, salonen2014impact}. These substrates are then fermented by intermediate organisms, largely within Firmicutes such as \textit{Dorea} and \textit{Fusicatenibacter}, producing acetate and other metabolic intermediates \citep{duncan2007reduced}. Previous studies have found an association between the \textit{Fusicatenibacter} genus and SCFAs~\citep{medawar2021gut, bartsch2025microbiota}; our analysis suggests that this association may be driven by finer resolution of the species level by \textit{Fusicatenibacter saccharivorans}. Finally, specialized bacteria producing butyrates, including \textit{Anaerostipes hadrus}, \textit{Eubacterium rectale}, and \textit{Roseburia}, convert these intermediates into butyrate, a key molecule for gut barrier function and immune regulation \citep{liu2024anaerostipeshadrus}. Actinobacteria phylum is known to be key in the development and maintenance of intestinal homeostasis and contributes to the function of the intestinal barrier through the production of acetate and lactate \citep{binda2018actinobacteria}.  Together, this system operates as an interdependent food web where no single organism completes the pathway, but instead SCFA production emerges from coordinated action across multiple taxa.  Our selection of across phylum to species highlights our procedure's ability to identify informative features at varying levels of resolution.

\begin{table}[H]
\centering
\footnotesize
\caption{Top 20 taxa ranked by learned $\hat\gamma$ coefficients.}
\centering
\begin{singlespace}
\begin{tabular}[t]{r|llr}
\toprule
Rank & Taxon & Level & $\hat \gamma_v \times  \var(\tilde X_{i,v})$\\
\midrule
1 & Bacteroidetes & Phylum & 11.87\\
2 & \textit{Anaerostipes hadrus} & Species & 11.86\\
3 & Bacteroidales & Order & 10.80\\
4 & \textit{Dorea} & Genus & 6.44\\
5 & Firmicutes & Phylum & 5.43\\
6 & Bacteroidia & Class & 3.64\\
7 & Oscillospiraceae & Family & 3.08\\
8 & Coriobacteriia & Class & 2.68\\
9 & \textit{Oscillibacter} & Genus & 1.76\\
10 & \textit{Fusicatenibacter saccharivorans} & Species & 1.72\\
11 & \textit{Agathobaculum butyriciproducens} & Species & 1.66\\
12 & \textit{Parabacteroides} & Genus & 1.65\\
13 & Tannerellaceae & Family & 1.54\\
14 & \textit{Eubacterium rectale} & Species & 0.86\\
15 & \textit{Roseburia intestinalis} & Species & 0.85\\
16 & \textit{Roseburia} & Genus & 0.80\\
17 & Actinobacteria & Phylum & 0.75\\
18 & \textit{Parabacteroides merdae} & Species & 0.25\\
19 & Actinobacteria & Class & 0.20\\
20 & \textit{Oscillibacter sp\_57\_20} & Species & 0.18\\
\bottomrule
\end{tabular}
\end{singlespace}
\label{tab:scfa}
\end{table}

As discussed previously, our procedure does not explicitly enforce the constraint that $v \in \hat \I$ implies that $\de(v) \not \in \hat \I$, and in finite samples, selected taxa may correspond to nested taxonomic levels. This happens for a handful of the selected taxa in our analysis; e.g., \textit{Roseburia} and \textit{Roseburia intestinalis}. The simultaneous appearance of multiple levels within the same lineage may point to differential associations which could be scientifically interesting.  
However, depending on the scientific focus, one could also post-process the selected model by collapsing selected features to the highest or lowest selected taxonomic node.

%The species \textit{Dorea longicatena} is a relatively new taxon within the human gut microbiome, first described in 2002 \citep{taras2002reclassification}. Although information linking this specific species to short-chain fatty acids is sparse, several \textit{Dorea} species are recognized producers of SCFAs (including acetate and occasionally butyrate) through carbohydrate fermentation, suggesting that it may also participate in similar metabolic pathways \citep{shin2025microbiota, song2025altered}. This, in turn, helps justify the similar $\hat{\gamma}$ value assigned to its genus, \textit{Dorea}, in \Cref{tab:scfa}.

%, because biological associations may manifest at different hierarchical depths, both fine‑scale (species) and coarse‑scale (family, order, class) features can differentially retain non‑zero importance in finite samples.  

%This optionality is offered in our \texttt{R} package. Notably, this flexibility is biologically well‑motivated; for example, increased \textit{Bifidobacterium} genus abundance (rather than individual species abundance) is well-known to correlate positively with plasma SCFA levels \citep{cui2024bifidobacterium, kalnina2023variations}. Overall, these results underscore the ability of KR-TEXAS to identify multi‑level taxonomic signals underlying SCFA metabolism, revealing both known associations and previously understudied lineages of potential biological relevance.

\section{Discussion}\label{sec:discussion}
We have proposed a method for nonparametric regression when the covariates may be aggregated at varying levels of granularity according to a known hierarchical structure. By using coordinate specific adaptive weights, our method consistently selects a target aggregation set which is of intrinsic interest in scientific problems. In addition, the model over the aggregated variables may be sparser than a model which only considers the original features; thus, selecting this target aggregation set also yields substantial statistical efficiency in the nonparametric setting. In simulations, we see that the proposed procedure outperforms existing methods in both prediction accuracy and model selection performance. Finally, we demonstrate the utility of our method by predicting short chain fatty acids with measurements of the microbiome which are hierarchically structured in a taxonomic tree. 

Future work may extend our method to the high-dimensional setting where $p$ is much larger than $n$. In addition, investigating control of the false discovery rate or false splitting rate~\citep{shao2025controlling} would be fruitful. Finally, future work may explore other penalties such as the exclusive lasso~\citep{campbell2017within} which more directly enforce the desired tree structure.

\section*{Acknowledgments}
The authors gratefully acknowledge Saurabh Mehta, Samantha Huey, and David Ruppert for their support and scientific expertise in gut microbiome analysis and nutrition. 
This work was supported by the National Institutes of Health under grant number T32HD113301: Artificial Intelligence and Precision Nutrition Training Program, Cornell University.

\section*{Data Availability Statement}
The method is implemented as the function \texttt{krtexas\_fit()} in the \texttt{krtexas} package which can be found at \url{https://github.com/sithijamanage/krtexas}. The data used in \Cref{sec:analysis} is publicly available in the \texttt{curatedMetagenomicData} \texttt{R} package.

\section*{Disclosure Statement}

The authors report there are no competing interests to declare.

\clearpage % 
\bibliographystyle{alpha}
\bibliography{compositional}

\clearpage % ensures Appendix starts on page after Bibliography
 \clearpage
\appendix
\crefalias{section}{appendix}
\crefalias{subsection}{appendix}

\part*{Appendix} % container above \section in article
\addcontentsline{toc}{part}{Appendix} % optional: show in main ToC

\begingroup
\etocsettocstyle{%
  \section*{Table of Contents}%
  \vspace{-0.25em}%
  \noindent\rule{\linewidth}{0.4pt}\par
  \vspace{0.75em}%
}{}
\localtableofcontents
\endgroup
\clearpage

\section{Conditional Bias and Variance for Epanechnikov Kernel Formulation}
\label{app:BiasVar}

In this section, we derive the conditional bias and variance of our estimator when using the Epanechnikov kernel, for fixed $\bm \gamma$.
The Epanechnikov kernel on aggregated features with inverse bandwidth parameters $\bm \gamma$ is 
\begin{equation}
\begin{aligned}
    K_{\g}(\tilde X_i,\tilde X_j) &= \left(1 - \sum_{v=1}^T \gamma_v (\tilde X_{i,v} - \tilde X_{j,v})^2 \right) \mathbbm{1}\left\{\sum_{v=1}^T \gamma_v (\tilde X_{i,v} - \tilde X_{j,v})^2 \leq 1\right\}.
\end{aligned}
\end{equation}

For ease of notation, we define $\bm A_\gamma = (\sqrt{\gamma}\mathbf{1}^\top)\odot \bm A$.

\begin{lemma}[Conditional Bias \& Variance of KR-TEXAS]
\label{lem:distFreeCondBiasVar}
\hfill

Assume (see \citealp{ruppert1994multivariate}):
\begin{itemize}
\label{ass:RW}
\item[\textbf{\textit{A1}}] 
(Interior point and smoothness)

Let $f_X$ be the common density of the $\mathbb{R}^p-$valued predictor variables.
The evaluation point $X_0$ lies in the interior of $\operatorname{supp}(f_X)$. At $X_0$, the conditional variance
\(\sigma(X_0):=\operatorname{Var}(Y\mid X=X_0)\) is continuous, $f$ is continuously differentiable, and all second-order derivatives of $m$ are continuous. Also, $f_X(X_0)>0$ and $\sigma(X_0)>0.$

\item[\textbf{\textit{A2}}]  (Bandwidth condition)

 As $n\to\infty$, each entry of $(\bm A_\gamma^\top \bm A_\gamma)$ and $n^{-1}\det(\bm A_\gamma^\top \bm A_\gamma)$ tends to zero, and 
the condition number of $(\bm A_\gamma^\top \bm A_\gamma)$ is uniformly bounded. That is to say, there exists a fixed constant $L$ such that the ratio of the largest eigenvalue of $(\bm A_\gamma^\top \bm A_\gamma)$ to the smallest eigenvalue of $(\bm A_\gamma^\top \bm A_\gamma)$ is at most $L$ for all $n$.
\end{itemize}

Then,
\begin{align*}
    \text{Bias} \left\{ \hat{m}_{h}(X_0) \mid X_1, \ldots, X_p \right\}
&\approx \frac{ \nabla m(X_0)^\top (\bm A_\gamma^\top\bm A_\gamma)^{-2}\nabla f_X(X_0) }{ (p+4)f_{X}(X_0) }
+ \frac{\text{tr} \left\{ (\bm A_\gamma^\top\bm A_\gamma)^\top \mathcal{H}_m(X_0) \bm A_\gamma^\top\bm A_\gamma \right\}}{2(p+4)}\\
    \text{Var} \left\{ \hat{m}_{h}(X_0) \mid X_1, \ldots, X_p \right\} 
&\approx  n^{-1}\det(\bm A_\gamma^\top\bm A_\gamma) \left( \frac{ p(p+2)\Gamma(p/2) }{ \pi^{p/2}(p+4) } \right) \frac{ \sigma(X_0) }{ f_{X}(X_0) }.\\
\end{align*}

\end{lemma}

Note that this approximation is to second order.  

\begin{proof}

To apply Theorem 2.1 in \citep{ruppert1994multivariate}, we must define the following expressions.

\subsubsection{Squared Norm of Kernel}
We denote the radial-symmetric (Epanechnikov) kernel by $\mathcal{K},$ as done in \citep{hardle2004nonparametric}.

Let
\[\mathcal{K}(u) = \frac{p(p+2)\Gamma(p/2)}{4\pi^{p/2}} (1 - \|u\|^2) \mathbbm{1}\{\|u\| \leq 1\}. \]

Then

\begin{align*}
    \|\mathcal{K}\|_2^2 
    &= \int \mathcal{K}(u)^2 du \\
    &= \left(\frac{p(p+2)\Gamma(p/2)}{4\pi^{p/2}}\right)^2 \int \left(1 - \|u\|^2\right)^2 \mathbbm{1}\left\{ \|u\| \leq 1 \right\} du \\
    &= \left(\frac{p(p+2)\Gamma(p/2)}{4\pi^{p/2}}\right)^2 
            \int_{0}^1 (1 - r^2)^2 r^{p-1} dr \int_{S^{p-1}} d\nu(v) \\
    &= \left(\frac{p(p+2)\Gamma(p/2)}{4\pi^{p/2}}\right)^2 
        \left( \int_0^1 (1 - r^2)^2 r^{p-1} dr \right) \left( \frac{2\pi^{p/2}}{\Gamma(p/2)} \right) \\
    &= \left(\frac{p(p+2)\Gamma(p/2)}{4\pi^{p/2}}\right)^2 
        \left( \frac{1}{p} - \frac{2}{p+2} + \frac{1}{p+4} \right) \left( \frac{2\pi^{p/2}}{\Gamma(p/2)} \right) \\
    &= \frac{ p^2(p+2)^2\Gamma(p/2)^2 }{ 16 \pi^p } 
        \left( \frac{1}{p} - \frac{2}{p+2} + \frac{1}{p+4} \right) 
        \left( \frac{2\pi^{p/2}}{\Gamma(p/2)} \right) \\
    &= \frac{ p^2(p+2)^2\Gamma(p/2)^2 }{ 16 \pi^p } 
        \cdot \frac{2\pi^{p/2}}{\Gamma(p/2)} 
        \cdot \frac{8}{p(p+2)(p+4)} \\
    &= \frac{ p(p+2)\Gamma(p/2) }{ \pi^{p/2}(p+4) }.
\end{align*}

\subsubsection{Second Moment of Epanechnikov Kernel}
\hfill

We wish to compute the second moment matrix
\[
M_2 = \int_{\mathbb{R}^p} u u^\top\, \mathcal{K}(u)\, du,
\]
where
\[
\mathcal{K}(u) = \frac{p(p+2)\Gamma(p/2)}{4\pi^{p/2}} (1 - \|u\|^2) \mathbbm{1}\{\|u\| \leq 1\}.
\]

By symmetry, $M_2$ is a scalar multiple of the identity:
\[
M_2 = \mu_2(\mathcal{K}) I_p,
\]
where
\[
\mu_2(\mathcal{K}) = \int_{\mathbb{R}^p} u_1^2\, \mathcal{K}(u)\, du.
\]

We can compute $\mu_2(\mathcal{K})$ as follows:

\begin{align*}
\mu_2(\mathcal{K}) &= \int_{\|u\|\leq 1} u_1^2\, \frac{p(p+2)\Gamma(p/2)}{4\pi^{p/2}} (1-\|u\|^2) du \\
       &= c \int_{\|u\|\le 1} u_1^2 (1-\|u\|^2) du,
\end{align*}
where $c = \frac{p(p+2)\Gamma(p/2)}{4\pi^{p/2}}$.

Let $u = r v$ with $r \in [0,1]$, $v \in S^{p-1}$ (the unit sphere). The Lebesgue measure transforms as $du = r^{p-1} dr\, d\nu(v)$, where $d\nu(v)$ is the surface measure on $S^{p-1}$. Also, $u_1^2 = r^2 v_1^2$ and $\|u\|^2 = r^2$.

Therefore,
\begin{align*}
\mu_2(\mathcal{K}) 
&= c \int_{0}^1 \int_{S^{p-1}} r^2 v_1^2 (1 - r^2)\, r^{p-1} d\nu(v) dr \\
&= c \int_{S^{p-1}} v_1^2 d\nu(v)\int_{0}^1 r^{p+1}(1 - r^2) dr.
\end{align*}

By symmetry, on the unit sphere,
\[
\int_{S^{p-1}} v_1^2 d\nu(v) = \frac{1}{p} \int_{S^{p-1}} \|v\|^2 d\nu(v) = \frac{1}{p} \mathrm{Area}(S^{p-1}),
\]
where
\[
\mathrm{Area}(S^{p-1}) = \frac{2\pi^{p/2}}{\Gamma(p/2)}.
\]
Therefore,
\[
\int_{S^{p-1}} v_1^2 d\nu(v) = \frac{2\pi^{p/2}}{p\,\Gamma(p/2)}.
\]

\begin{align*}
\int_0^1 r^{p+1}(1 - r^2) dr 
&= \int_0^1 r^{p+1} dr - \int_0^1 r^{p+3} dr \\
&= \left. \frac{r^{p+2}}{p+2} \right|_0^1 - \left. \frac{r^{p+4}}{p+4} \right|_0^1 \\
&= \frac{1}{p+2} - \frac{1}{p+4}.
\end{align*}

Then
\begin{align*}
\mu_2(\mathcal{K})&= c \cdot \frac{2\pi^{p/2}}{p\, \Gamma(p/2)} \left(\frac{1}{p+2} - \frac{1}{p+4}\right) \\
       &= \frac{p(p+2)\Gamma(p/2)}{4\pi^{p/2}} \cdot \frac{2\pi^{p/2}}{p\, \Gamma(p/2)} \left(\frac{1}{p+2} - \frac{1}{p+4}\right) \\
       &= \frac{p+2}{2} \left(\frac{1}{p+2} - \frac{1}{p+4}\right) \\
       &= \frac{1}{2}\left(\frac{p+4-(p+2)}{p+4}\right) \\
       &= \frac{1}{2} \cdot \frac{2}{p+4} \\&= \frac{1}{p+4}.
\end{align*}

\subsubsection{Final Distribution-Free Conditional Bias and Variance}
\hfill

We must slightly modify the proof of Theorem 2.1 in \citep{ruppert1994multivariate} to achieve the desired result.
$H^{1/2}$ denotes the bandwidth matrix. We show below that in  our setting, $H = (\bm A_\gamma^\top\bm A_\gamma)^{-1}:$
\begin{align*}
    \|H^{-1/2}u\|_2^2 &= \|(\bm A_\gamma^\top\bm A_\gamma)^{1/2}u\|_2^2\\
    &= \|(\bm A_\gamma^\top\bm A_\gamma)^{1/2}u\|_2^2\\
    &= u^\top((\bm A_\gamma^\top\bm A_\gamma)^{1/2})^\top(\bm A_\gamma^\top\bm A_\gamma)^{1/2}u\\
    &=  u^\top(\bm A_\gamma^\top\bm A_\gamma)^{1/2}(\bm A_\gamma^\top\bm A_\gamma)^{1/2}u\\
    &=  u^\top(\bm A_\gamma^\top\bm A_\gamma)u\\
    &= \|\bm A_\gamma u\|_2^2\\
\end{align*}
as needed.

Now we apply Theorem 2.1 in \citep{ruppert1994multivariate} (NW version given in Section 4.5.1 in \citep{hardle2004nonparametric}) as 

\[
\text{Bias} \left\{ \hat{m}_{h}(X_0) \mid X_1, \ldots, X_p \right\} 
\approx \mu_2(K) \frac{ \nabla m(X_0)^\top H H^\top \nabla f_X(X_0) }{ f_{X}(X_0) }
+ \frac{1}{2} \mu_2(K) \, \text{tr} \left\{ H^\top \mathcal{H}_m(X_0) H \right\},
\]

\[
\text{Var} \left\{ \hat{m}_{H}(X_0) \mid X_1, \ldots, X_p\right\} 
\approx \frac{1}{n \, \det(H)} \| \mathcal{K} \|_2^2 \, \frac{ \sigma(x) }{ f_{X}(x) },
\]
in the interior of the support of $f_{X}$.
% Noting that in our case, $\mathbf{H}  = hI_\tau,$ we arrive at the desired result.

Thus, 
\begin{align*}
    \text{Bias} \left\{ \hat{m}_{h}(X_0) \mid X_1, \ldots, X_p \right\}
&\approx \frac{ \nabla m(X_0)^\top (\bm A_\gamma^\top\bm A_\gamma)^{-2}\nabla f_X(X_0) }{ (p+4)f_{X}(X_0) }
+ \frac{\text{tr} \left\{ (\bm A_\gamma^\top\bm A_\gamma)^\top \mathcal{H}_m(X_0) \bm A_\gamma^\top\bm A_\gamma \right\}}{2(p+4)}\\
    \text{Var} \left\{ \hat{m}_{h}(X_0) \mid X_1, \ldots, X_p \right\} 
&\approx  n^{-1}\det(\bm A_\gamma^\top\bm A_\gamma) \left( \frac{ p(p+2)\Gamma(p/2) }{ \pi^{p/2}(p+4) } \right) \frac{ \sigma(X_0) }{ f_{X}(X_0) }.\\
\end{align*}

\end{proof}

\newpage

\section{Proofs of Model Selection Consistency}
\label{app:selectionConsistency}
At times throughout the proofs, we will use $L_n(\g)$ to denote $L_n(\g; \tilX, \Y)$ for notational brevity. 

\begin{lemma}\label{lem:boundedDerivative}
Under \Cref{assump:bounded}, the leave-one-out cross validation loss,
\begin{equation}
     L_n(\g; \tilX, \Y) = \frac{1}{n} 
\sum_{i=1}^{n}\left(Y_i - \hat m_{-i}(\tilde X_i; \tilX, \Y, \gamma) \right)^2,
\end{equation}
has uniformly bounded gradients. That is, for all $k\in \mathcal{T}$,
\begin{equation}\label{eq:bounded_grad}
\begin{aligned}
\left \vert \frac{\partial L_n(\bm X, \Y; \g)}{\partial \gamma_{k}} \right\vert&\leq 16 B^4. 
\end{aligned}
\end{equation}
\begin{proof}
    
First, note that
\begin{equation}
\begin{aligned}
   \frac{\partial L_n(\g; \tilX, \Y)}{\partial \gamma_{k}} = \frac{1}{n}\sum_{i=1}^n 2\left(Y_i - \hat{m}_{-i}(\tilde X_i; \tilX,\Y, \g) \right)\frac{\partial \hat{m}_{-i}(\tilde X_i; \tilX,\Y, \g) }{\partial \gamma_{k}},
\end{aligned}
\end{equation}
and let $d_{ij,v} = (\tilde X_{iv} - \tilde X_{jv})$ and $\omega_{ij} = \exp(-\sum_{v=1}^T \g_v d_{ij,v}^2)$. Then, $\omega_{ij}' := \frac{\partial \omega_{ij}}{\partial \gamma_{k}} = - d_{ij,k}^2 \omega_{ij}$ so that
\begin{equation}
\nonumber
\begin{aligned}
\frac{\partial \hat{m}_{-i}(X_i; \tilX,\Y,\g)}{\partial \gamma_{k}} &=  \frac{\partial}{\partial \gamma_{k}} \frac{\sum_{j\neq i} \omega_{ij} Y_j }{\sum_{j\neq i} \omega_{ij}} =  \frac{(\sum_{j\neq i} {\omega}'_{ij} Y_j) (\sum_{j\neq i} \omega_{ij}) - (\sum_{j\neq i} \omega'_{ij})(\sum_{j\neq i} \omega_{ij} Y_j))  }{(\sum_{j\neq i} \omega_{ij})^2}\\
 &= \frac{1}{\sum_{j\neq i} \omega_{ij}}\left(\sum_{j\neq i} \omega'_{ij} Y_j -  (\sum_{j\neq i} \omega'_{ij})\frac{\sum_{j\neq i} \omega_{ij} Y_j}{\sum_{j\neq i} \omega_{ij}} \right)\\
  &= \frac{1}{\sum_{j\neq i} \omega_{ij}}\left(\sum_{j\neq i} \omega'_{ij}\left[ Y_j -  \frac{\sum_{\ell\neq i} \omega_{i\ell} Y_\ell}{\sum_{\ell\neq i} \omega_{i\ell}}\right] \right)\\
   &= \frac{1}{\sum_{j\neq i} \omega_{ij}}\left(\sum_{j\neq i} \omega'_{ij}\left[ Y_j -  \hat{m}_{-i}(X_i; \tilX,\Y,\g)\right] \right) \\
&=-\frac{\sum_{j\neq i}  \omega_{ij}\left[ Y_j - \hat{m}_{-i}(X_i; \tilX,\Y,\g)\right]d_{ij,k}^2}{\sum_{j\neq i} \omega_{ij}}.
\end{aligned}
\end{equation}
By Assumption \ref{assump:bounded} we have

\begin{equation}
\begin{aligned}
 \left\vert \frac{\partial \hat{m}_{-i}(X_i; \tilX,\Y,\g)}{\partial \gamma_{k}} \right\vert &= 
\left\vert \frac{\sum_{j\neq i}  \omega_{ij}\left[ Y_j - \hat{m}_{-i}(X_i; \tilX,\Y,\g)\right]d_{ij,k}^2}{\sum_{j\neq i} \omega_{ij}}\right\vert\\ 
&\leq 2B \left\vert \frac{\sum_{j\neq i}  w_{ij}d_{ij,k}^2}{\sum_{j\neq i} w_{ij}}\right\vert \leq 8 B^3.
\end{aligned}
\end{equation}
Thus, since $\vert(Y_i - \hat{m}_{-i}(\tilde X_i; \tilX,\Y, \g) \vert < 2B$, we have
\begin{equation}\label{eq:bounded_grad}
\begin{aligned}
\left \vert \frac{\partial L_n(\bm X, \Y; \g)}{\partial \gamma_{k}} \right\vert&\leq 16 B^4. 
\end{aligned}
\end{equation}

\end{proof}
    
\end{lemma}

\noFalsePos*

\begin{proof}
Recall that KR-TEXAS solves the following problem: 
\begin{equation}\label{eq:kr_restate}
\hat{\g} = \arg\min_{\g}L_n(\g; \tilX, \Y)
 + \lambda_n \sum_{v=1}^T \hat w_v\gamma_v,
 \end{equation}
where  
\begin{equation}\small
 L_n(\g; \tilX, \Y) = \frac{1}{n} 
\sum_{i=1}^{n}\left(Y_i - \hat m_{-i}(\tilde X_i; \tilX, \Y, \gamma) \right)^2 \text { and } 
     \hat m_{-i}(\tilde X_i; \tilX, \Y, \gamma)  = \frac{\sum_{j\neq i} \K(\tilde X_j, \tilde X_i) Y_j}{\sum_{j\neq i}\K(\tilde X_j, \tilde X_i)}.
 \end{equation}

If $\gammadag$ is a stationary point of \Cref{eq:kr_restate} and $\gammadag_{k} \neq 0$, the subgradient (KKT) condition that must be satisfied is
\[
 \frac{\partial L_n(\bm X, \Y; \gammadag)}{\partial \gamma_{k}} = -\lambda_n \hat w_{k}.
\]

However, if $k \not \in \I^\star$, then by Assumption \ref{assump:tuning} one of the three following conditions must hold:
\begin{enumerate}
    \item There exists $j,k \in \leaves(k')$ such that $\partial_j m \neq \partial_k m$, so
    \begin{itemize}
        \item $C_{k,1} = \Omega(1)$, and thus $\hat{w}_{k'} = \Omega(n^{a_2b})$.
    \end{itemize}
\item $\partial_j m = \partial_k m$ for all $j,k \in  \leaves(k')$ and $\partial_j m = 0$ for all $j \in \rm leaves(k')$, so
    \begin{itemize}
        \item $C_{k',2} = O(n^{-a_1})$, and thus $\hat{w}_{k'} = \Omega(n^{a_1b})$. 
    \end{itemize}
        \item $\partial_j m = \partial_k m \neq 0$ for all $j,k \in \leaves(k')$ and $\partial_j m = \partial_k s$ for all $j \in \leaves(k')$ and $s \in \bigcup_{k \in \text{sib}(t)}\leaves(k)$, so
            \begin{itemize}
        \item $C_{k',3} = O(n^{-a_1})$, and thus $\hat{w}_{k'} = \Omega(n^{a_1b})$.  
    \end{itemize}
    \end{enumerate}
By Assumption \ref{assump:tuning} for large enough $n$, $a_2b - d > 0$ and $a_1b - d > 0$, so $\lambda_n \hat{w}_{k'} \xrightarrow{p} \infty$. Thus, by~\Cref{lem:boundedDerivative},
\begin{align*}
P\left(\exists \gammadag \text{ such that }  \gammadag_{k}>0 \text{ for } k \not \in \Istar\right) &\leq P\left(\exists k \not \in \Istar \text{ such that } \max_{\g} \frac{\partial L_n(\bm X, \Y; \bm \gamma)}{\partial \gamma_{k}} \geq \lambda_n \hat w_{k'}\right)\\
& \leq P\left(\exists k \not \in \Istar \text{ such that } 16B^4 \geq \lambda_n \hat w_{k'}\right)
\rightarrow 0.    
\end{align*}

\end{proof}

\subsection{Proof of \Cref{thm:selectionConsistency}}
To prove \Cref{thm:selectionConsistency}, we first establish several auxiliary lemmas. In most settings, we will fix $\g$. Thus, for notational convenience, let 
\begin{align*}
N_i(x) &= \sum_{j \neq i}\K(x, X_j) Y_j \\
D_i(x) &= \sum_{j \neq i}\K(x, X_j) \\
D_i &= D_i(X_i)\\
N_i &= N_i(X_i)\\
\phi(x) &= \E(\K(x, X))\\
\psi(x) &= \E(\K(x, X)Y)\\
\mu_i &= \E(D_i \mid X_i) = (n-1)\phi(X_i)\\
\nu_i &= \E(N_i \mid X_i) = (n-1)\psi(X_i)\\
L_i &= \left( Y_i - N_i/D_i\right)^2 = \left( Y_i - \hat m_{-i}(X_i)\right)^2 \\
L_n &= \frac{1}{n}\sum_{i = 1}^n L_i = \frac{1}{n}\sum_{i = 1}^n \left( Y_i - N_i/D_i\right)^2\\
\g^{(-1/2)} &= \prod_{v}(\gamma_v^{-1/2} + 1_{\{\gamma_v = 0\}}).
\end{align*}

\begin{lemma}\label{lem:goodEventWHP}
Under \Cref{assump:bounded}, for any $\epsilon \in (0, 1)$ and fixed $\g$, 
\begin{equation}
P\left( \bigcup_i \left\{\vert D_i - \mu_i\vert \geq \epsilon \mu_i \cup \vert N_i - \nu_i\vert \geq B\mu_i \epsilon\right\} \right) \leq 4n \exp(-3^{-1}c_f\epsilon^2 n\g^{(-1/2)}).
\end{equation}

\end{lemma}
\begin{proof}
For each $i$, conditional on $X_i$, $D_i = \sum_{j\neq i}K(X_i, X_j)$ is the sum of i.i.d variables in $(0,1]$. Furthermore, $\var(K(X_i, X_j) \mid X_i) \leq \E(K(X_i, X_j)^2 \mid X_i) \leq \E(K(X_i, X_j) \mid X_i) = \phi(X_i)$. Thus, $\var(D_i \mid X_i) \leq \mu_i$. Letting $Z_j = K(X_i, X_j) - \E(\K(X_i, X_j)\mid X_i)$, we also have $\var(Z_j\mid X_i) \leq \phi(X_i)$ and $\var(D_i \mid X_i) \leq \mu_i$. Thus, by Bernstein's inequality we have 
\begin{align*}
P\left(\vert D_i - \mu_i \vert \geq \epsilon\mu_i \mid X_i \right) 
&
\leq 2\exp\left(-\frac{(\epsilon\mu_i)^2}{2\mu_i + 2/3\epsilon\mu_i}\right)\\
&\leq 2\exp\left(-\frac{\epsilon^2\mu_i}{2 + 2/3\epsilon}\right)\\ 
&\leq 2\exp\left(-3^{-1}c_f\epsilon^2 n\g^{(-1/2)}\right),
\end{align*}
where the last inequality comes from~\Cref{assump:bounded}. 

Analogously, note that conditional on $X_i$, $N_i = \sum_{j\neq i}Y_jK(X_i, X_j)$ is the sum of i.i.d variables in $[-B,B]$. Furthermore, $\var(K(X_i, X_j)Y_j \mid X_i) \leq \E(B^2K(X_i, X_j)^2 \mid X_i) \leq B^2\E(K(X_i, X_j) \mid X_i) = B^2\phi(X_i)$. Thus, $\var(N_i \mid X_i) \leq B^2\mu_i$. Letting $Z_j = Y_jK(X_i, X_j) - \E(Y_j\K(X_i, X_j)\mid X_i)$, we also have $\var(Z_j\mid X_i) \leq B^2\phi(X_i)$ and $\var(N_i \mid X_i) \leq B^2\mu_i$. Again, by Bernstein's inequality we have 
\begin{align*}
P\left(\vert N_i - \nu_i \vert \geq \epsilon B\mu_i \mid X_i \right) 
&
\leq 2\exp\left(-\frac{(B\epsilon\mu_i)^2}{2B^2\mu_i + 2/3B^2\epsilon\mu_i}\right)\\
&\leq 2\exp\left(-\frac{\epsilon^2\mu_i}{2 + 2/3\epsilon}\right)\\ 
&\leq 2\exp\left(-3^{-1}c_f\epsilon^2 n\g^{(-1/2)}\right).
\end{align*}

Since the RHS of both inequalities do not depend on $X_i$, we use a union bound over the $2n$ events to conclude that 
 \begin{equation}
P\left( \bigcup_i \left\{\vert D_i - \mu_i\vert \geq \epsilon \mu_i \cup \vert N_i - \nu_i\vert \geq B\mu_i \epsilon\right\} \right) \leq 4n \exp(-3^{-1}c_f\epsilon^2 n\g^{(-1/2)}).
\end{equation}
\end{proof}

\begin{lemma}\label{lem:loocv_concentration}
Under \Cref{assump:bounded}, for any fixed $\g$ and $\epsilon \in (0, 1/2)$, with probability at least $1 - 4n \exp(-3^{-1}c_f\epsilon^2 n\g^{(-1/2)})$ we have:
\begin{equation}
\vert L_n - \E(L_n) \vert < c_B( \epsilon + n\exp(-3^{-1}c_f\epsilon^2n\g^{(-1/2)}))
\end{equation}
for some constant $c_B$ which depends only on $B$.
When $n\g^{(-1/2)} > \frac{3}{c_f\epsilon^2}\log(n/\epsilon)$, this simplifies to 
\begin{equation}
\vert L_n - \E(L_n) \vert < c_B\epsilon.
\end{equation}

\end{lemma}
\begin{proof}
Let $\Delta_i := (\Delta N_i, \Delta D_i)$ where 
\begin{align*}
\Delta N_i &:= N_i - \nu_i,   \\
\Delta D_i &:= D_i - \mu_i. \\
\end{align*}
Let $\ell_i(n, d; Y) = (Y - n/d)^2$; the partial derivatives are:
\begin{align*}
    \partial\ell_i/ \partial n &= \frac{-2}{d}(Y - n/d)\\
    \partial\ell_i/ \partial d &= \frac{2n}{d^2}(Y - n/d).
\end{align*}
\begin{comment}
\begin{align*}
    \partial\ell_i/ \partial n &= \frac{-2}{d}(Y - n/d)\\
    \partial\ell_i/ \partial d &= \frac{2n}{d^2}(Y - n/d).
%    \partial^2\ell_i/ \partial n^2 &= \frac{2}{d^2}\\
%    \partial^2\ell_i/ \partial d^2 &= \frac{6n^2}{d^4} - \frac{4nY}{d^3}\\
%    \partial^2\ell_i/ \partial n \partial d &= \frac{2}{d^2}(Y - 2n/d).\\
\end{align*}
\end{comment}
We now bound each of the terms, conditional on the event for $\epsilon \in (0, 1/2)$:
\[\mathcal{A}_\epsilon = \bigcap_i \left\{\vert D_i - \mu_i\vert \leq \epsilon \mu_i \cap \vert N_i - \nu_i\vert \leq B\mu_i \epsilon\right\}\]
which, by Lemma~\ref{lem:goodEventWHP}, occurs with probability at least $1 - 4n \exp(-3^{-1}c_f\epsilon^2 n\g^{(-1/2)})$.

Let $(n_t, d_t):= (\nu_i, \mu_i) + t \Delta_i$ for some $t \in [0,1]$. First note that:
\begin{align*}
d_t &\in [(1-\epsilon)\mu_i, (1+\epsilon)\mu_i]\\
\left\vert \frac{n_t}{d_t} \right\vert &\leq \frac{\max(\vert \nu_i\vert ,\vert N_i\vert)}{d_t}\leq \frac{B\mu_i(1+\epsilon)}{\mu_i(1-\epsilon)} = B\frac{1+\varepsilon}{1-\varepsilon}\leq 3B\\
\vert Y - n_t/d_t \vert &\leq \vert Y \vert + \vert n_t / d_t \vert \leq B + B\frac{1+\varepsilon}{1-\varepsilon}\leq 4B.  
\end{align*}
Then,
\begin{align*}
    \vert \partial\ell_i/ \partial n \vert  &= \vert \frac{-2}{d}(Y - n/d) \vert \leq C_{B}\mu_i^{-1}\\
     \vert \partial\ell_i/ \partial d \vert  &= \vert \frac{2n}{d^2}(Y - n/d)\vert\leq C_{B}\mu_i^{-1},
\end{align*}
so that
\begin{equation}\label{eq:taylor_remainder}
    \vert \nabla\ell_i(n_t, d_t)^T \Delta_i \vert \leq C_{B} \epsilon.
\end{equation}

\begin{comment}
\begin{align*}
    \vert \partial\ell_i/ \partial n \vert  &= \vert \frac{-2}{d}(Y - n/d) \vert \leq C_{B}\mu_i^{-1}\\
     \vert \partial\ell_i/ \partial d \vert  &= \vert \frac{2n}{d^2}(Y - n/d)\vert\leq C_{B}\mu_i^{-1} \\
 \vert \partial^2\ell_i/ \partial n^2 \vert&= \vert\frac{2}{d^2} \vert \leq \frac{2}{(1-\epsilon)^2\mu_i^2}\\
   \vert \partial^2\ell_i/ \partial d^2 \vert &= \vert\frac{6n^2}{d^4} - \frac{4nY}{d^3}\vert \leq 6\left(B\frac{1+\epsilon}{1-\epsilon}\right)\frac{1}{(1-\epsilon)^2\mu_i^2} + 4B\left(B\frac{1+\epsilon}{1-\epsilon}\right)\frac{1}{(1-\epsilon)^2\mu_i^2} \leq C_{B, \epsilon}\frac{1}{\mu_i^2} \\
   \vert \partial^2\ell_i/ \partial n \partial d\vert &= \vert\frac{2}{d^2}(Y - 2n/d)\vert \leq \frac{2(B+ B(1+\epsilon)/(1-\epsilon)}{(1-\epsilon)^2\mu_i^2}.\\
\end{align*}
%Thus, every element of $\nabla\ell_i(\nu_i, \mu_i)$ is bounded above in absolute value by $C_{B, \epsilon}\mu_i^{-1}$ and every element of $H(n_t, d_t)$ is bounded in absolute value by $C_{B, \epsilon}\mu_i^{-2}$.
\end{comment}

Then the Taylor expansion for $L_i = \ell_i(N_i, D_i)$ around $\mu_i, \nu_i$ is:
\begin{align*}
L_i &=\ell_i(\nu_i, \mu_i) + \nabla\ell_i(n_t, d_t)^T \Delta_i %+ \frac{1}{2}\Delta_i^T H((\nu_i, \mu,i) + \theta_i\Delta_i)\Delta_i.
\end{align*}
for some $t \in (0,1)$. Thus,
\begin{equation}\label{eq:li_to_exp}
\vert L_i - \E(L_i \mid X_i, Y_i) \vert \leq \vert\ell_i(\nu_i, \mu_i) - \E(L_i(N_i, D_i)\mid X_i, Y_i) \vert  + C_{B} \epsilon. 
\end{equation}
To bound the first term, note that conditional on $X_i, Y_i$, we have
\begin{equation}\label{eq:l_nu_mu_to_exp}
    \begin{aligned}
    \vert \E(L_i \mid X_i, Y_i) -\ell_i(\nu_i, \mu_i) \vert &= \left\vert \E\left(\left[L_i -\ell_i(\nu_i, \mu_i)\right]1_{\mathcal{A}_\epsilon} + \left[L_i -\ell_i(\nu_i, \mu_i)\right]1_{\mathcal{A}^C_\epsilon}  \mid X_i, Y_i\right)  \right\vert \\
    &\leq \left\vert \E\left(\left[L_i -\ell_i(\nu_i, \mu_i)\right]1_{\mathcal{A}_\epsilon}  \mid X_i, Y_i\right) \right \vert\\
    &\quad+ \left\vert \E\left(\left[L_i-\ell_i(\nu_i, \mu_i)\right]1_{\mathcal{A}^C_\epsilon}  \mid X_i, Y_i\right) \right \vert\\
    &\leq C_B \epsilon + 4B^2P(\mathcal{A}^C_\epsilon) \\
    &\leq C_B( \epsilon + n\exp(-3^{-1}c_f\epsilon^2n\g^{(-1/2)})). 
    \end{aligned}
\end{equation}
Combining \Cref{eq:li_to_exp} and \Cref{eq:l_nu_mu_to_exp}, we have on the event $\mathcal{A}_\epsilon$ that
\begin{equation}\label{eq:li_bound}
\vert \ell_i(N_i, D_i) - \E(L_i(N_i, D_i)) \vert \leq  C_B( \epsilon + n\exp(-3^{-1}c_f\epsilon^2n\g^{(-1/2)})).
\end{equation}
Because $L_n$ is the average of $L_i$, then also
\begin{equation}\label{eq:li_bound}
\vert L_n - \E(L_n) \vert < C_B( \epsilon + n\exp(-3^{-1}c_f\epsilon^2n\g^{(-1/2)})).
\end{equation}

\end{proof}

\begin{restatable}{lemma}{unifControl}
\label{lem:unifControl}
Under \Cref{assump:bounded} and~\ref{assump:contFunc}, with probability $1-o(1)$, we have simultaneously for all $\g \in G_n$
\begin{equation}
\left\vert L_n(\g) - \E[L_n(\g)] \right\vert \lesssim n^{-\beta/(2\beta+\vert \g \vert_0)}\log(n).
\end{equation}
This implies, for $\tilde \g =\arg\min_{\g \in G_n}L_n(\g)$, we have
\begin{equation}
\left\vert L_n(\tilde \g) - \E[L_n(\tilde \g)] \right\vert \lesssim n^{-\beta/(2\beta+\vert \tilde \g \vert_0)}\log(n) < n^{-\beta/(2\beta+T)}\log(n).
\end{equation}
\end{restatable}
\begin{proof}
Fix some $\I \subseteq \mathcal{T}$, let $\vert \I \vert = s$, and let
\[G_{\I,n} = \{\g :\; 0 \leq \gamma_v \leq C n^{2/(2\beta+s)}\;\; \forall v \in \I \quad \text{ and } \quad \gamma_v = 0 \;\; \forall v \not \in \I\}.\]
By Lemma~\ref{lem:boundedDerivative} the derivative of $L_n(\gamma)$ is bounded by $16B^4$ so for fixed data, it is Lipschitz in $\g \in G_{\I, n}$ (with respect to $\vert \cdot \vert_2$) with constant $16B^4\sqrt{s}$; i.e.,
\[\vert L_n(\g) - L_n(\g')\vert \leq 16B^4\sqrt{s} \vert\g -\g' \vert_2. \]

Let $\varepsilon = n^{-\beta/(2\beta+s)} / (32B^4\sqrt{s})$ and $H_{\I, n}(\varepsilon)$ denote an $\varepsilon$-net of $G_{\I,n}$ (with respect to $\vert \cdot \vert_2$) so that 
\begin{equation}\label{eq:epsNetOracle}
\begin{aligned}
\max_{\g \in G_{\I, n}} \vert L_n(\g) - \E(L_n(\g))\vert &\leq \max_{\g \in H_{\I, n}(\varepsilon)} \vert L_n(\g) - \E(L_n(\g)) \vert + n^{-\beta/(2\beta+s)}.
\end{aligned}
\end{equation}
We now compute $\vert H_{\I,n}(\varepsilon) \vert$, the covering number for $G_{\I, n}$. Since $G_{\I,n}$ is a cube in $\mathbb{R}^{s}$, then the covering number for $G_{\I,n}$ is of order $(\sqrt{s}C_gn^{2/(2\beta +s)}/\varepsilon)^{s} \lesssim s^{s/2}C_g^s n^{2s/(2\beta+s)}n^{\beta s/(2\beta +s)} \lesssim n^{(2+\beta)s/(2\beta+s)}$. 

We then apply~\Cref{lem:loocv_concentration} with deviation $\zeta = n^{-\beta/(2\beta+s)}\log(n)$. Note that $\min_{\g \in G_{\I, n}}n\g^{(-1/2)} \asymp n^{1-2s/2(2\beta + s)} = n^{2\beta/(2\beta + s)}$.  Thus, for $n$ large enough, the inequality $ n^{2\beta/(2\beta+s)} > \frac{3}{c_f\zeta^2}\log(n/\zeta) = \frac{3n^{2\beta/(2\beta+s)}}{c_f\log^2(n)}\log(n^{(3\beta+s)/(2\beta+s)}\log(n))$ will be satisfied so that:
\begin{equation}\label{eq:mainBoundOracle}\small
\begin{aligned}
        P\left(\max_{\g \in  H_{\I,n}(\varepsilon)}\right.&\left.\vert L_n(\g)- \E\left[L_n( \g)\right] \vert \geq n^{-\beta/(2\beta+ s)}\log(n) \vphantom{\max_{\g \in  H_{\I,n}(\varepsilon)}} \right) \\
        &\leq \vert H_{\I,n}(\varepsilon)\vert \max_{\g \in  H_{\I,n}(\varepsilon)}4n\exp(-3^{-1}c_f (n^{-\beta/(2\beta+ s)}\log(n))^2 n\g^{(-1/2)})\\
        &\lesssim  n^{(2+\beta)s/(2\beta+s) + 1}\max_{\g \in  H_{\I,n}(\varepsilon)}\exp(-3^{-1}c_f (n^{-\beta/(2\beta+ s)}\log(n))^2 n\g^{(-1/2)})\\
        &\lesssim n^{s(3+\beta)/(2\beta+s)}\exp(- \log^2(n)).
\end{aligned}
\end{equation}
The second to last inequality comes from the fact that $\min_{\g \in H_{\I,n}}n\g^{(-1/2)} >n^{1-2s/(2(s+2\beta))} = n^{2\beta/(2\beta+s)}$. Combining \Cref{eq:epsNetOracle} and \Cref{eq:mainBoundOracle} implies that 
\begin{equation}\label{eq:epsNetFinal_oracle}
\begin{aligned}
\max_{\g \in G_{\I, n}} \vert L_n(\g) - \E(L_n(\g))\vert \geq n^{-\beta/(2\beta+ s)}\log(n)
\end{aligned}
\end{equation}
with probability $O(n^{s(3+\beta)/(2\beta+s)}\exp(- \log^2(n))$.

Note that $G_n = \cup_{\I \subseteq \mathcal{T}} G_{\I, n}$ and there are $2^T$ subsets $\I \subseteq T$. Thus, by taking a union bound over all $\I$, we have
\begin{equation}\label{eq:mainBoundOracleUnion}\small
\begin{aligned}
        P\left(\max_{\I \subseteq \mathcal{T}}\max_{\g \in  G_{\I,n}}\right.&\left.\vert L_n(\g)- \E\left[L_n( \g)\right] \vert \geq n^{-\beta/(2\beta+ s)}\log(n) \vphantom{\max_{\g \in  H_{\I,n}(\varepsilon)}} \right) \\
        &\lesssim 2^T\max_{\I \subseteq \mathcal{T}}\left[ n^{\vert \I\vert(3+\beta)/(2\beta+\vert \I\vert)}\exp(- \log^2(n))\right]\\
        &\leq 2^T \left[ n^{T(3+\beta)/(2\beta+1)}\exp(- \log^2(n))\right].
\end{aligned}
\end{equation}

Since $T$ is fixed, the RHS of~\Cref{eq:mainBoundOracleUnion} is $o(1)$. This yields the desired result: we have with probability at $1-o(1)$ that for all $\g \in G_{n}$,
\begin{equation}\label{eq:epsNetFinal_oracle}
\begin{aligned}
\vert L_n(\g) - \E(L_n(\g))\vert \lesssim n^{-\beta/(2\beta+ \vert \g \vert )}\log(n). 
\end{aligned}
\end{equation}
    
\end{proof}

\modelSelect*
\begin{proof}
Let $Q_n(\g) = L_n(\g; \tilX, \Y) + \lambda_n \sum_v\hat w_v \g_v$. We now show that the global optimum
\begin{equation}
\arg\min_{\g\in G_{n}} Q_n(\g) = \hat \gamma^\star \in G_{n} \setminus \check G_{n}
\end{equation}
with probability $1- o(1)$. Specifically, we show that with probability $1-o(1)$, there exists some $\g^{\dagger,n} \in G_{n} \setminus \check G_{n}$ such that $Q_n(\g^\dagger) < \sigma^2 + \delta/2$ and that $\min_{\g \in \check G_{n}}Q_n(\g) \geq \sigma^2 + \delta/2$. Thus, the global minimizer must be in $G_{n} \setminus \check G_{n}$.

\textbf{Existence of $\g^{\dagger,n} \in G_{n} \setminus \check G_{n}$ such that $Q_n(\g^{\dagger,n}) < \sigma^2 + \delta/2$.}

For any $k \in \Istar$, we have $\partial_j m \neq 0$ and $\partial_j m = \partial_j' m$  for all $j,j' \in \leaves(k)$
and there exists $s \in \bigcup_{j \in \sib(k)}\leaves(j)$ such that 
$\partial_j m \neq \partial_s m$ for some $j \in \leaves(k)$. Thus, $C_{k,3} = O(n^{-a_1})$, $C_{k,2} = O(1)$, and $C_{k,1} = O(1)$  so $\hat{w}_{k} = O(n^{(a_2-a_1)b})$ and $\lambda_n \hat w_k = O(n^{(a_2-a_1)b - d})$. 
Let $\g^{\dagger,n}$ denote a sequence of parameters where $\g^{\dagger,n}_k = n^r$ for some
$  0 <r < \min(2 / (2\beta + \vert \Istar\vert), d - \left(a_2 - a_1\right)b) $ if
$k \in \Istar$ and $\g^{\dagger,n}_{k'} = 0$ if $k' \in \mathcal{T} \setminus \Istar$. Thus,
$\lambda_n \sum_{v}\hat{w}_{v}\gamma^{\dagger, n}_v \asymp \vert \Istar\vert n^{(a_2-a_1)b-d +r}$.
Since $(a_2-a_1)b-d +r <0$,
the regularization term $\lambda_n \sum_{t=1}^T\hat{w}_{t,n}\gamma^\dagger_{t,n}\rightarrow 0$ as $n\rightarrow \infty$. Therefore, for large enough $n$, with probability $1-o(1)$ we have $\lambda_n \sum_{k}\hat{w}_{k,n}\gamma^{\dagger, n}_k \leq \delta/16$.

Under~\Cref{assump:contFunc}, $\lim_{n\rightarrow \infty} \E(L_n(\g^{\dagger,n})) = \E((Y - \E(Y_i\mid P_{\Istar}(X_i)))^2) = \sigma^2$, so for large enough $n$ we have that $\E(L_n(\g^{\dagger,n})) \leq \sigma^2 + \delta/16$. Then, by~\Cref{lem:unifControl}, we have with probability $1-o(1)$ that
\begin{equation}
\begin{aligned}
 Q_n(\g^{\dagger,n}) &= \E(L_n(\g^{\dagger,n})) +  (L_n(\g^{\dagger,n}) - \E(L_n(\g^{\dagger,n}))) + \lambda_n\sum_{k}\hat{w}_{k}\gamma^{\dagger, n}_k \\
 &\leq \sigma^2 + \delta/16 + \delta/4 + \delta/16 < \sigma^2 + \delta/2.         
\end{aligned}
\end{equation}

\textbf{With probability $1-o(1)$, we have $\min_{\g \in \check G_n}Q_n(\g) \geq \sigma^{2} + \delta/2$}.

For any $\check{\g} \in \check G_n$ let $\I = \{v : \check \g_v \neq 0\}$. Then,
\begin{equation}
    \begin{aligned}
        \E\left[(L_n(\check \g)\right] &= \E([Y_1 - \hat m_{-1}(X_1; \check{\g})]^2)\\
        &=\E_{\{X_i, Y_i\}_{i \neq 1}}\left[\E_{X_1, Y_1}\left([Y_1 - \hat m_{-1}(P_\mathcal{I}(X_1); \check{\g})]^2 \mid \{X_i, Y_i\}_{i \neq 1} \right)  \right]\\
         &\leq\E_{\{X_i, Y_i\}_{i \neq 1}}\left[\E_{X_1, Y_1}\left([Y_1 - \E(Y_1 \mid \check P_\mathcal{I}(X_1))]^2 \mid \{X_i, Y_i\}_{i \neq 1} \right)  \right]\\
        &= \E_{X_1, Y_1}([Y_1 - \E(Y_1 \mid \mathcal{P}_\I(X_1))]^2) \\
        &\geq \sigma^2 + \delta.\\
    \end{aligned}
\end{equation}

The first equality comes from symmetry of $L_n$. The second line can be seen by noting $\check m_{-1}(X_1; \check g)$ only depends on $\mathcal{P}_\I(X_1)$. The third lines comes from the property of the conditional mean, the penultimate line comes from the fact that each observation is independent, and the last inequality comes from \Cref{assump:seperation}.

Thus, by \Cref{lem:unifControl}, we have with probability $1-o(1)$, we have $\max_{\g \in \check G_{n}} \vert L_n(\check \g) - \E\left[(L_n(\check \g)\right] \vert < \delta/2$. Thus,
\begin{align}
    \min_{\g \in \check G_{n}} Q_n(\g) &\geq \min_{\g \in \check G_{n}} \E\left[(L_n(\check \g)\right] - \max_{\g \in \check G_{n}} \vert L_n(\check \g) - \E\left[(L_n(\check \g)\right] \vert \geq \sigma^2 - \delta/2.
    \end{align}

Thus $\hat \g \in G_{n} \setminus \check G_{n}$ so $\hat \I$ is not a strict non-super-model of $\Istar$ with probability $1- o(1)$. Combined with~\Cref{lem:noFalsePos} this implies $\hat \I = \Istar$. 
\end{proof}

\section{Implementation Details}
\label{app:implementation}

\begin{algorithm}[H]
\small
\caption{Details for KR-TEXAS implementation}
\label{alg:krtexas-practical}
\begin{algorithmic}[1]
\Require Data $(\bm X,Y)$, tree matrix $\bm A$, tuning grid $\Lambda$ for $\lambda_n$;
         numbers of random restarts $R_1$ (Step A), $R_2$ (Step B);
         choice of pilot method (\texttt{NW\_ML}, \texttt{LLR}, or \texttt{LQR});
         choice of distance for penalties ($\ell_2$ or $\ell_1$)
\Ensure Final vector $\hat{\bm \gamma}$, selected $\hat\lambda_n$
\State \textbf{Step A: Adaptive Weight Construction}
\State \textbf{Leaf–level metric learning at $\lambda=0$}
\For{$r = 1,\ldots,R_1$}  \Comment{random restarts}
  \State Initialize $\bm \gamma^{(r)}$ using a randomized strategy (e.g. ``large'' vs ``small'' scale)
  \State Run L--BFGS--B to (approximately) solve \eqref{eq:kr_texas_opt} with
         $\bm A = \bm I_p$, $\lambda_n = 0$, and $\hat {\bm w}=1$.
  \State Record the solution $\bm \gamma^{(r)}_0$ and its training loss
\EndFor
\State Let $\check{\bm \gamma}^{\text{leaf}}$ be the best $\bm {\gamma}^{(r)}_0$ (smallest training loss)
\State \textbf{Oversmoothing and interior points}
\State Set oversmoothed metric $\tilde{\bm \gamma}^{\text{leaf}} \gets (\check{\bm \gamma}^{\text{leaf}})^z$ with $z\in(0,1)$.
\State For each $i\in[n]$, compute kernel neighborhood mass
       $K_i = \sum_{j\neq i} K_{\tilde{\bm \gamma}^{\text{leaf}}}(X_i,X_j)$
\State Let $\mathcal{J}$ be the set of $m = \lfloor n/10\rfloor$ indices with the largest $K_i$
       (interior points)
\State \textbf{Pilot derivative estimation on leaves}
\If{pilot method = \texttt{LLR} or \texttt{LQR}}
  \For{each $i\in[n]$}
    \State Solve the local linear problem \eqref{eq:leaf_opt_llr} or local quadratic problem
           centered at $X_i$ with kernel $K_{\tilde{\bm \gamma}^{\text{leaf}}}$ to obtain
           $\hat\beta_i = (\hat\beta_{i,1},\ldots,\hat\beta_{i,p})$
           as an estimate of $\{\partial_v m(X_i)\}_{v=1}^p$
  \EndFor
\ElsIf{pilot method = \texttt{NW\_ML}}
  \State Use the KR--TEXAS fit with metric $\tilde{\bm \gamma}^{\text{leaf}}$ to compute
         gradients $\widehat{\partial_v m}(X_i)$ w.r.t. leaf features
\EndIf
\end{algorithmic}
\end{algorithm}

\begin{algorithm}[H]
\caption{Details for KR-TEXAS implementation (Continued)}
\begin{algorithmic}[1]
    \State \textbf{Adaptive penalty components $C_{v,1},C_{v,2},C_{v,3}$}
\For{each tree node $v \in \mathcal{T}$}
  \State Use only interior points $i\in\mathcal{J}$ to compute
         \begin{align*}
           C_{v,1} &\gets \binom{|\leaves(v)|}{2}^{-1}
             \sum_{u,w\in\leaves(v)} \frac{1}{|\mathcal{J}|}
             \sum_{i\in\mathcal{J}} d\!\left(\widehat{\partial_u m}(X_i), \widehat{\partial_w m}(X_i)\right),\\
           C_{v,2} &\gets \frac{1}{|\leaves(v)|}
             \sum_{u\in\leaves(v)} \frac{1}{|\mathcal{J}|}
             \sum_{i\in\mathcal{J}} d\!\left(\widehat{\partial_u m}(X_i), 0\right),\\
           C_{v,3} &\gets \frac{1}{|\leaves(v)|\,|\bigcup_{k\in\sib(v)}\leaves(k)|}
             \sum_{u\in\leaves(v)}\sum_{k\in\sib(v)}\sum_{w\in\leaves(k)}
             \frac{1}{|\mathcal{J}|}\sum_{i\in\mathcal{J}}
             d\!\left(\widehat{\partial_u m}(X_i),\widehat{\partial_w m}(X_i)\right),
         \end{align*}
         where $d(a,b)=|a-b|$ for $\ell_1$ or $d(a,b)=(a-b)^2$ for $\ell_2$.
    \State $w_v \gets (n^{a_2} C_{v,1})^b + C_{v,2}^{-b} + C_{v,3}^{-b}$
\EndFor

\State \textbf{Step B: Penalized fitting with adaptive weights} 
\State Split $[n]$ into $k = 1\ldots K$ folds with train/test sets denoted as $D_{train}^{(k)}$ and $D_{test}^{(k)}$ 
\For{$m \in [R_2]$}
\For{$\lambda \in \Lambda$ and $k = 1:K$}
        \State Use random initialization then warm starts to get a local optimum $\hat \g_{\lambda,m}^{(k)}$ of
        \[
        %\mathcal{L}(\g; \tilX_{D_{train,m}^{(k)}}, \Y_{D_{train,m}^{(k)}},\lambda)= 
        \frac{1}{\vert D_{train}^{(k)}\vert}\sum_{i\in D_{train}^{(k)}} (Y_i - \hat m_{-i}(\tilde X_i; \tilX_{D_{train}^{(k)}}, \Y_{D_{train}^{(k)}}, \g))^2
        + \lambda \sum_{v=1}^T \hat w_v \g_v.
        \]
        using L--BFGS--B with the spred factorization $\bm \gamma = \bm u\odot \bm w$ and constraint $\bm (\bm u, \bm w) \geq 0$
\EndFor
\EndFor
\State Select $(\lambda^\star, m^\star) = \arg \min_{\lambda \in \Lambda, m} \sum_{k=1}^K \sum_{i \in D_{test,m}^{(k)}}(Y_i - \hat m_{-i}(\tilde X_i; \tilX_{D_{train}^{(k)}}, \Y_{D_{train}^{(k)}}, \hat\g_{\lambda, m}^{(k)}))^2$

\State Use $m^\star$ random initialization to get local optimum $\hat \g$ of
\[
        \mathcal{L}(\g; \tilX, \Y,\lambda^\star)
        = \frac{1}{n}\sum_{i\in [n]} (Y_i - \hat m_{-i}(\tilde X_i; \tilX, \Y, \g))^2
        + \lambda^\star \sum_{v=1}^T \hat w_v \g_v.
\]
\State \Return $(\hat{\bm\gamma},\lambda^*)$ 
\end{algorithmic}
\end{algorithm}

We now describe the full set of input parameters for the \texttt{krtexas\_fit} function.
\begin{itemize}
    \item \{\texttt{X, Y, A}\}. \texttt{X, Y,} and \texttt{A} correspond to $\bm X, Y,$ and $\bm A$ as defined in Section \ref{sec:methods:model}.  
    \item  \{\texttt{nfolds, lambda, min\_lambda, max\_lambda, nlambda}\}. \texttt{nfolds} determines the number of cross-validation folds for Step B in \Cref{alg:krtexas_fit_overall}. The user can set the \texttt{lambda} parameter to fit KR-TEXAS (\Cref{eq:kr_texas_opt}) with a fixed, predetermined $\lambda$. If \texttt{lambda} is not specified, the user can specify the range over which the algorithm will use \texttt{nfolds}-cross-validation to find the $\lambda$ value with the lowest corresponding loss by specifying \texttt{min\_lambda, max\_lambda,} and \texttt{nlambda}. \texttt{min\_lambda} and \texttt{max\_lambda} determine the (inclusive) lower and upper bound of the cross-validation range while \texttt{nlambda} specifies how many values are evaluated within that range.
    \item \{\texttt{eps}\}. \texttt{eps} is the epsilon for convergence test in LBFGS \citep{lbfgs2022}. It sets the convergence criterion for KR-TEXAS as in line 5 of \Cref{alg:krtexas_fit_overall}.
    \item \{\texttt{silent, parallel, n\_cores, warm\_start}\}. The \texttt{silent} parameter silences the console output from the LBFGS optimization. If the \texttt{parallel} parameter is set to \texttt{TRUE}, the cross validation folds are run in parallel, with the number of cores used specified by \texttt{n\_cores}. \texttt{warm\_start} is a logical parameter that allows for the user to specify whether to use ``warm starts'' in the cross-validation step to find a suitable $\lambda$. When set to \texttt{TRUE}, the average of the $\hat{\bm \gamma}$ values from the previous candidate $\lambda$ value are used to initialize the optimizations for the current candidate $\lambda$.
    \item \{\texttt{method, gamma\_init\_strat}\}. The \texttt{method} parameter takes on values of \texttt{``NW\_ML'', ``LLR''}, and \texttt{``LQR''}, which select NW with metric learning, local linear regression, and local quadratic regression as the pilot estimator, respectively (line 1 in \Cref{alg:krtexas_fit_overall}). \texttt{gamma\_init\_strat} initializes each element of $u$ and $w$ (defined in \Cref{sec:spred}) with $N(\mu = 1, \sigma = 1/4)$ if set to ``small'', $N(\mu = \max (1, n^{2/(4+T)}), \sigma = 1)$ if set to ``large'', and $N(\mu = 0.1, \sigma = 0.01)$ if set to ``smallest''.
    \item \{\texttt{distance, num\_restarts\_stage\_1, num\_restarts\_stage\_2}\}. In \Cref{eq:weights}, the adaptive weights are defined using $\ell_2$ distance, which corresponds to \texttt{distance = "L2"}, but we also provide optionality for using $\ell_1$ distance by selecting \texttt{distance = "L1"}. In step 1 of \Cref{alg:krtexas_fit_overall}, \Cref{eq:leaf_opt} is solved \texttt{num\_restarts\_stage\_1} times, with $1/3$ ``smallest'' initializations, $1/3$ ``small'' initializations, and $1/3$ ``large'' initializations. The solution vector with the lowest loss among all restarts is used in the adaptive weight construction. Similarly, \texttt{num\_restarts\_stage\_2} determines the number of random restarts conducted with the selected $\lambda$ value to solve \Cref{eq:kr_texas_opt} (line 8 of \Cref{alg:krtexas_fit_overall}). $1/3$ of the initializations are ``smallest'', $1/3$ are ``small'', and $1/3$ initializations are ``large''.
    Furthermore, in the software implementation, we set tuning parameters in \Cref{eq:weights} such that $a_2 \gets 1/(2(2+p))$, and $b\gets 1$.
    \item \{\texttt{max\_attempts\_stage\_3}\} In the case that the optimization in line 14 of \Cref{alg:krtexas_fit_overall} does not converge, it is retried up to \texttt{max\_attempts\_stage\_3} times using a random initialization.
\end{itemize}

\subsubsection{\textbf{spred implementation}}
\label{sec:spred}
 Instead of directly optimizing the $\gamma$ vector, we use a factorized representation proposed by~\citep{ziyin2023spred}:
\[\bm \gamma = u \odot w,\]
where $u$ and $w$ are vectors of length $T$, and $\odot$ denotes Hadamard (element-wise) multiplication.

\begin{itemize}
\item The optimization is performed over the combined parameter vector $\theta = [u, w]$ of length $2T$, rather than directly over $\bm \gamma$ of length $T$.

\item The spred penalty is applied to both $u$ and $w$ components:
\[
\text{spred penalty} = \kappa(\|u\|_2^2 + \|w\|_2^2),
\]
where $\kappa = \lambda/2$.

\item The gradients with respect to the original parameters are computed using the chain rule:
\begin{align*}
\frac{\partial L}{\partial u} &= \frac{\partial \mathcal{L}}{\partial \gamma} \odot w \\
\frac{\partial \mathcal{L}}{\partial w} &= \frac{\partial \mathcal{L}}{\partial \gamma} \odot u.
\end{align*}

% \item The \texttt{Matrix} package in R is used to efficiently compute internal node values by converting the matrix \texttt{A} into a sparse matrix representation: \texttt{A\_sparse <- Matrix(A, sparse = TRUE)}. The product is then computed as \texttt{Ax <- as.matrix(A\_sparse \%*\% t(X))}.
\end{itemize}

% The benefit of implementing $spred$ rather than optimizing over $\gamma$ directly is that the factorization can help avoid numerical issues that arise from direct optimization of $\gamma$.

\newpage

\section{Details for Numerical Experiments}
\label{app:NumericalExpDetails}

The following parameters are used in \Cref{sec:numerics}: \texttt{nfolds=5}, \texttt{nlambda=10}, \texttt{eps=1e-6}, \texttt{num\_restarts\_stage\_1 = 30}, \texttt{num\_restarts\_stage\_2 = 30}, \texttt{max\_attempts\_stage\_3 = 10}. For RARE, we use the settings \texttt{lam.min.ratio = 1e-6}, \texttt{nlam = 20}, \texttt{nalpha = 10}, \texttt{rho=0.01}, \texttt{eps1=1e-6}, \texttt{eps2=1e-6}, and \texttt{maxite=1e4}. The lasso utilizes \texttt{nfolds=5}, and is fit using the \texttt{glmnet} R package. The bandwidth for the Nadaraya–Watson kernel regression is chosen via cross-validated grid search for 

$\sigma \in \{0.001, 0.005, 0.01, 0.02, 0.05, 0.1, 0.2, 0.5, 1, 10, 10^2, 10^3, 10^4, 10^5, 10^6\}$.

For a one-to-one comparison, we use the $\gamma$ parameter in RARE to represent variable selection. If $\gamma_t$ is non-zero, then we denote variable $t$ as selected. These $\gamma$ parameters dictate the aggregation level of the final $\beta$ coefficients, and this gives RARE the most fair comparison when comparing to the method proposed in this work. 

% In the SCFA data analysis Section \ref{sec:analysis}, the parameters for the implementation of KR-TEXAS are as follows: 
% \texttt{nfolds}=4, \texttt{nlambda}=10, \texttt{eps}=1e-8, \texttt{gamma\_init\_strat} = ``small''.
The prediction performance of all methods across all sample size and covariance settings is given in \Cref{fig:rmse_all}.

\begin{figure}
    \centering
\includegraphics[width=\linewidth]{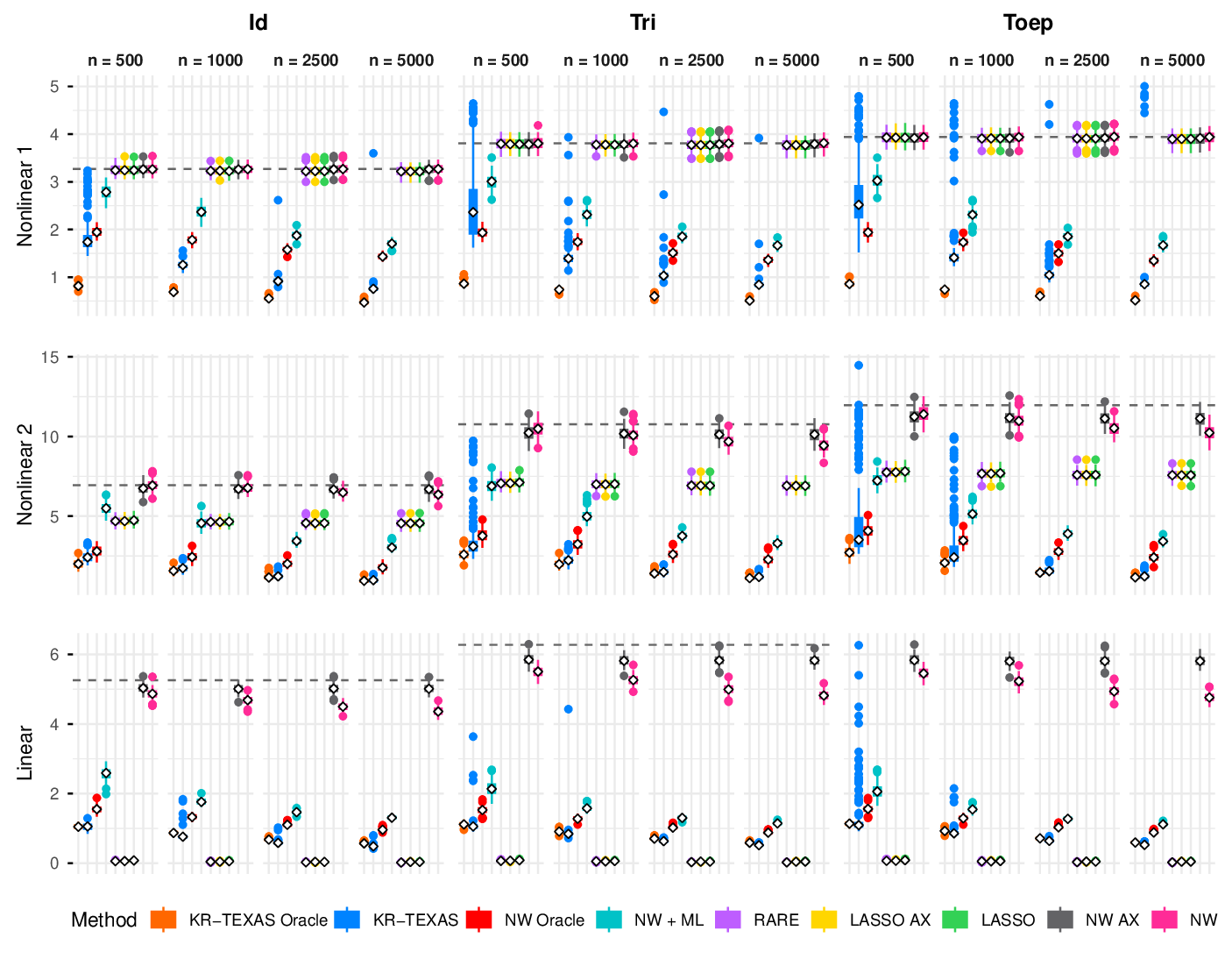}
    \caption{Prediction performance of all methods across covariance and sample size settings.}
    \label{fig:rmse_all}
\end{figure}

The variable selection performance of all methods across all sample size and covariance settings is given in \Cref{tab:allCasesTab}.

\begin{table}[ht]
\caption{Variable Selection: All Cases}
\centering
\begin{tabular}[t]{ccc|@{\,}r@{\,\,}r@{\,\,}r@{\,\,}r@{\,}|@{\,}r@{\,\,}r@{\,\,}r@{\,\,}r@{\,}|@{\,}r@{\,\,}r@{\,\,}r@{\,\,}r@{\,}|@{\,}r@{\,\,}r@{\,\,}r@{\,\,}r@{\,}}
& & & \multicolumn{4}{c}{$n=500$} & \multicolumn{4}{c}{$n=1000$} & \multicolumn{4}{c}{$n=2500$} & \multicolumn{4}{c}{$n=5000$} \\
& & & SN & SP & Prec & NPV & SN & SP & Prec & NPV & SN & SP & Prec & NPV & SN & SP & Prec & NPV \\
\midrule
\multirow{9}{*}{\rotatebox{90}{Nonlinear 1}} & \multirow{3}{*}{\rotatebox{90}{Id}} & KRT & \textbf{1.00} & .35 & .03 & \textbf{1.00} & \textbf{1.00} & .61 & .05 & \textbf{1.00} & \textbf{1.00} & .72 & .07 & \textbf{1.00} & \textbf{1.00} & .87 & .13 & \textbf{1.00} \\
\cline{3-19}
 &  & LASSO & .19 & \textbf{.98} & \textbf{.15} & .98 & .22 & \textbf{.98} & \textbf{.20} & .98 & .23 & \textbf{.98} & \textbf{.18} & .98 & .23 & \textbf{.98} & \textbf{.17} & .98 \\
\cline{3-19}
 &  & RARE & .17 & .97 & .11 & .98 & .22 & .97 & .14 & .98 & .22 & .97 & .14 & .98 & .22 & .97 & .12 & .98 \\
\cline{2-19}
 & \multirow{3}{*}{\rotatebox{90}{Tri}} & KRT & \textbf{.85} & .56 & .04 & \textbf{.99} & \textbf{1.00} & .63 & .05 & \textbf{1.00} & \textbf{1.00} & .66 & .06 & \textbf{1.00} & \textbf{1.00} & .93 & \textbf{.22} & \textbf{1.00} \\
\cline{3-19}
 &  & LASSO & .15 & \textbf{.98} & \textbf{.14} & .98 & .21 & \textbf{.98} & \textbf{.18} & .98 & .23 & \textbf{.98} & \textbf{.17} & .98 & .22 & \textbf{.98} & .17 & .98 \\
\cline{3-19}
 &  & RARE & .14 & .97 & .09 & .98 & .20 & \textbf{.98} & .15 & .98 & .22 & .97 & .13 & .98 & .22 & \textbf{.98} & .17 & .98 \\
\cline{2-19}
 & \multirow{3}{*}{\rotatebox{90}{Toep}} & KRT & \textbf{.77} & .62 & .04 & \textbf{.99} & \textbf{1.00} & .59 & .05 & \textbf{1.00} & \textbf{1.00} & .69 & .06 & \textbf{1.00} & \textbf{1.00} & .90 & .17 & \textbf{1.00} \\
\cline{3-19}
 &  & LASSO & .15 & \textbf{.98} & \textbf{.13} & .98 & .20 & \textbf{.98} & \textbf{.17} & .98 & .22 & \textbf{.98} & \textbf{.17} & .98 & .22 & \textbf{.98} & \textbf{.18} & .98 \\
\cline{3-19}
 &  & RARE & .14 & .97 & .10 & .98 & .18 & \textbf{.98} & .14 & .98 & .21 & .97 & .12 & .98 & .21 & \textbf{.98} & .17 & .98 \\
\midrule
\multirow{9}{*}{\rotatebox{90}{Nonlinear 2}} & \multirow{3}{*}{\rotatebox{90}{Id}} & KRT & \textbf{.95} & .40 & .03 & \textbf{1.00} & \textbf{.92} & .53 & .04 & \textbf{1.00} & \textbf{.96} & .54 & .04 & \textbf{1.00} & \textbf{.98} & .55 & .04 & \textbf{1.00} \\
\cline{3-19}
 &  & LASSO & .59 & \textbf{.96} & \textbf{.23} & .99 & .65 & \textbf{.96} & \textbf{.23} & .99 & .69 & \textbf{.96} & \textbf{.25} & .99 & .74 & \textbf{.96} & \textbf{.26} & .99 \\
\cline{3-19}
 &  & RARE & .56 & \textbf{.96} & .21 & .99 & .65 & .94 & .18 & .99 & .69 & .94 & .19 & .99 & .75 & .94 & .21 & .99 \\
\cline{2-19}
 & \multirow{3}{*}{\rotatebox{90}{Tri}} & KRT & \textbf{.87} & .50 & .03 & \textbf{.99} & \textbf{.88} & .57 & .04 & \textbf{1.00} & \textbf{.86} & .63 & .04 & \textbf{1.00} & \textbf{.86} & .64 & .05 & \textbf{1.00} \\
\cline{3-19}
 &  & LASSO & .50 & \textbf{.96} & \textbf{.22} & \textbf{.99} & .55 & \textbf{.96} & \textbf{.22} & .99 & .63 & \textbf{.95} & \textbf{.20} & .99 & .66 & \textbf{.95} & \textbf{.19} & .99 \\
\cline{3-19}
 &  & RARE & .48 & \textbf{.96} & .18 & \textbf{.99} & .51 & .95 & .18 & .99 & .60 & .94 & .17 & .99 & .65 & .91 & .12 & .99 \\
\cline{2-19}
 & \multirow{3}{*}{\rotatebox{90}{Toep}} & KRT & \textbf{.88} & .45 & .03 & \textbf{.99} & \textbf{.86} & .54 & .04 & \textbf{.99} & \textbf{.87} & .63 & .05 & \textbf{1.00} & \textbf{.85} & .65 & .05 & \textbf{1.00} \\
\cline{3-19}
 &  & LASSO & .48 & \textbf{.96} & \textbf{.21} & \textbf{.99} & .53 & \textbf{.96} & \textbf{.22} & \textbf{.99} & .61 & \textbf{.95} & \textbf{.20} & .99 & .64 & \textbf{.95} & \textbf{.19} & .99 \\
\cline{3-19}
 &  & RARE & .47 & \textbf{.96} & .19 & \textbf{.99} & .49 & \textbf{.96} & .18 & \textbf{.99} & .59 & \textbf{.95} & .19 & .99 & .64 & .92 & .13 & .99 \\
\midrule
\multirow{9}{*}{\rotatebox{90}{Linear}} & \multirow{3}{*}{\rotatebox{90}{Id}} & KRT & \textbf{.81} & .60 & .04 & .99 & \textbf{.80} & .85 & .10 & \textbf{1.00} & \textbf{.80} & .70 & .05 & .99 & \textbf{.80} & .64 & .04 & .99 \\
\cline{3-19}
 &  & LASSO & .80 & \textbf{.94} & \textbf{.21} & \textbf{1.00} & \textbf{.80} & \textbf{.94} & \textbf{.20} & \textbf{1.00} & \textbf{.80} & \textbf{.99} & \textbf{.56} & \textbf{1.00} & \textbf{.80} & \textbf{.99} & \textbf{.75} & \textbf{1.00} \\
\cline{3-19}
 &  & RARE & .80 & .91 & .15 & \textbf{1.00} & \textbf{.80} & \textbf{.94} & \textbf{.20} & \textbf{1.00} & \textbf{.80} & .91 & .15 & \textbf{1.00} & \textbf{.80} & .94 & .22 & \textbf{1.00} \\
\cline{2-19}
 & \multirow{3}{*}{\rotatebox{90}{Tri}} & KRT & \textbf{.83} & .71 & .05 & \textbf{1.00} & \textbf{.81} & .81 & .08 & \textbf{1.00} & \textbf{.82} & .58 & .04 & .99 & \textbf{.81} & .53 & .03 & .99 \\
\cline{3-19}
 &  & LASSO & .80 & \textbf{.94} & \textbf{.21} & \textbf{1.00} & .80 & \textbf{.94} & \textbf{.20} & \textbf{1.00} & .80 & \textbf{.99} & \textbf{.55} & \textbf{1.00} & .80 & \textbf{.99} & \textbf{.71} & \textbf{1.00} \\
\cline{3-19}
 &  & RARE & .80 & .92 & .17 & \textbf{1.00} & .80 & \textbf{.94} & \textbf{.20} & \textbf{1.00} & .80 & .93 & .19 & \textbf{1.00} & .80 & .94 & .20 & \textbf{1.00} \\
\cline{2-19}
 & \multirow{3}{*}{\rotatebox{90}{Toep}} & KRT & \textbf{.82} & .72 & .06 & \textbf{1.00} & \textbf{.83} & .80 & .08 & \textbf{1.00} & \textbf{.82} & .57 & .04 & .99 & \textbf{.82} & .53 & .03 & .99 \\
\cline{3-19}
 &  & LASSO & .80 & \textbf{.94} & \textbf{.21} & \textbf{1.00} & .80 & \textbf{.94} & .20 & \textbf{1.00} & .80 & \textbf{.99} & \textbf{.53} & \textbf{1.00} & .80 & \textbf{.99} & \textbf{.70} & \textbf{1.00} \\
\cline{3-19}
 &  & RARE & .80 & .92 & .18 & \textbf{1.00} & .80 & \textbf{.94} & \textbf{.21} & \textbf{1.00} & .80 & .93 & .19 & \textbf{1.00} & .80 & .94 & .20 & \textbf{1.00} \\
\hline
\end{tabular}
\label{tab:allCasesTab}
\end{table} 

\subsection{Incorporating Additional Variables Not Present in the Tree}
\label{app:nonTreeVariables}

To incorporate a variable that is not included in the tree, we can simply add it as a leaf node by augmenting the $A$ matrix with a row and column of zeros and then setting the ($T+1,T+1$) element element to $1.$ This is illustrated below with a toy example.

\[\bm A_{\text{original}} = 
\begin{bmatrix}
1 & 0 & 0 \\
0 & 1 & 0 \\
0 & 0 & 1 \\
1 & 0 & 1
\end{bmatrix} \text{ and } \bm A_{\text{variable-added}} = 
\begin{bmatrix}
1 & 0 & 0 &0\\
0 & 1 & 0 &0\\
0 & 0 & 1 &0 \\
1 & 0 & 1 &0\\
0 & 0 & 0 & 1\\
\end{bmatrix}.
\]

Thus, adding $c$ non-tree variables can be expressed compactly by augmenting $\bm A_{\text{original}}$ as follows:

\[
\bm A_{\text{with-}c\text{-variables}} =
\begin{bmatrix}
\bm A_{\text{original}} & \bm 0 \\
\bm 0 & I_c
\end{bmatrix},
\]

where $\bm 0$ denotes a matrix of zeros of appropriate size, and $I_c$ is the $c \times c$ identity matrix corresponding to the newly added variables.

\FloatBarrier

\newpage
\section{Data Analysis details}
\label{app:dataAnalysis}

We use the \texttt{curatedMetagenomicData} R/Bioconductor package \citep{cmd2017} to obtain metagenomic profiles and associated metadata. Starting from the \texttt{sampleMetadata} object, we restrict to stool samples from the control group of studies:
\begin{itemize}
  \item \texttt{study\_condition == "control"},
  \item \texttt{body\_site == "stool"}.
\end{itemize}
For participants with repeated observations, we retain only the earliest visit by grouping on (\texttt{study\_name}, \texttt{subject\_id}) and selecting the smallest visit number. We then keep only studies with more than $1{,}000$ qualifying samples, which yields the AsnicarF~2021 and LifeLinesDeep~2016 studies \citep{asnicar2021microbiome, zhernakova2016population}. 

To construct short-chain fatty acid (SCFA) scores, we use HUMAnN3 \citep{humann3} pathway abundance profiles using \texttt{returnSamples(metadata, dataType = "pathway\_abundance")}. This yields a pathway-by-sample matrix, which we align to the sample order in the filtered metadata.

We filter the pathways by selecting rows whose names contain any of the following substrings:
\[
\text{``butyr''},\ \text{``butanoate''},\ \text{``propionat''},\ \text{``propanoate''},\ \text{``acetat''}
.\] For each sample, we compute the aggregate SCFA pathway abundance as the sum of the selected pathway abundances after adding a small pseudocount of $10^{-6}$ to avoid zeros:
\[
Y_{\text{SCFA},i} = \sum_{r \in \mathcal{P}_{\text{SCFA}}} \bigl( \text{PA}_{r i} + 10^{-6} \bigr),
\]
where $\text{PA}_{r i}$ denotes the HUMAnN3 pathway abundance for pathway $r$ and sample $i$, and $\mathcal{P}_{\text{SCFA}}$ is the set of SCFA-related pathways.

We then apply a log transform,
\[
Y^{\log}_{\text{SCFA},i} = \log\bigl(1 + Y_{\text{SCFA},i}\bigr),
\]
and standardize across samples to obtain a $z$-scored SCFA outcome,
\[
Y_i = \frac{Y^{\log}_{\text{SCFA},i} - \bar Y^{\log}_{\text{SCFA}}}{\widehat{\mathrm{sd}}(Y^{\log}_{\text{SCFA}})}.
\]
This standardized log-sum $Y_i$ is used as the response in the analysis.

For the covariates, we extract species-level relative abundance profiles using
\texttt{returnSamples(metadata, dataType = "relative\_abundance")}, which returns a feature-by-sample matrix of MetaPhlAn \citep{metaphlan2012} taxonomic profiles. We align columns to match the filtered metadata and restrict to species-level features by retaining only rows whose names contain the species pattern
\texttt{"(\^{}|\textbackslash|)s\_\_"},
corresponding to taxa annotated at the species level.

Samples with zero total species abundance are removed. For the remaining samples, we renormalize the species abundances to sum to one within each sample, so that they can be interpreted as species-level compositions:
\[
 X^\dagger_{j i} = \frac{X_{j i}}{\sum_{j'} X_{j' i}},
\]
where $X_{j i}$ denotes the original relative abundance of species $j$ in sample $i$.

We compute, for each species $j$, its prevalence
\[
\text{prev}_j = \frac{1}{n} \sum_{i=1}^n \mathbf{1}\{X^\dagger_{j i} > 0\}
\]
and mean abundance
\[
\overline{X^\dagger}_j = \frac{1}{n} \sum_{i=1}^n  X^\dagger_{j i}.
\]
Species are ranked first by decreasing prevalence, with ties broken by decreasing mean abundance. We select the top $K = 50$ species according to this ranking. Let $\mathcal{S}_{\text{top}} = \{j_1,...,j_K\}$ be the indices of $K$ species. The covariate matrix used for modeling is
\[
\bm X \in \mathbb{R}^{n \times K}, \qquad
X_{ik} = X^{\dagger\top}_{j_k, i}, \quad k = 1,\dots,K,
\]
where $\{j_1,\dots,j_K\} = \mathcal{S}_{\text{top}}$, and samples (rows) and species (columns) ordered so that the row names of $\bm X$ coincide with the sample IDs used for $Y$. 

We then construct the $\bm A$ matrix. For each of the $K$ species, we obtain its full lineage from kingdom through species and take the union over species to define the node set. If an internal node has exactly one child, we remove it and connect its parent directly to its child to avoid internal nodes with $1$ child. The resulting tree has $T$ total nodes (leaves plus internal nodes). As discussed in \Cref{sec:methods:model}, we encode the tree via the binary matrix
\[
\bm A \in \{0,1\}^{T \times K},
\]
where row $v$ corresponds to a node in the tree, column $j$ corresponds to a selected species, and $A_{v j} = 1$ if and only if species $j$ is a descendant of node $v$.

We use the following parameters for fitting KR-TEXAS to the SCFA data in \Cref{sec:analysis}: \texttt{max\_lambda = 0.15, nfolds = 4, eps = 1e-8, method = "LLR", gamma\_init\_strat = "small",
  num\_restarts\_stage\_1 = 10,
  num\_restarts\_stage\_2 = 10}. The value \texttt{max\_lambda = 0.15} is chosen as a valid upper bound on $\lambda$; 
at this value, $\hat{\bm{\gamma}} = \bm{0}$, so we avoid the computational expense of searching for such a
\texttt{max\_lambda}.

\begin{table}[H]
\centering
\caption{All 41 taxa selected by KR-TEXAS in SCFA analysis.}
\begin{minipage}{0.48\linewidth}
\centering

 \begin{table}[H]
\centering\begingroup\fontsize{8}{8}\selectfont
\begin{singlespace}
\begin{tabular}{cllp{2.2cm}}
\toprule
Rank & Taxon & Level & \textbf{$\hat{\gamma}\!\times\!\text{var}$}\\
\midrule
1 & Bacteroidetes & Phylum & 11.87\\
2 & \textit{Anaerostipes hadrus} & Species & 11.86\\
3 & Bacteroidales & Order & 10.80\\
4 & \textit{Dorea} & Genus & 6.44\\
5 & Firmicutes & Phylum & 5.43\\
6 & Bacteroidia & Class & 3.64\\
7 & Oscillospiraceae & Family & 3.08\\
8 & Coriobacteriia & Class & 2.68\\
9 & \textit{Oscillibacter} & Genus & 1.76\\
10 & \textit{Fusicatenibacter saccharivorans} & Species & 1.72\\
11 & \textit{Agathobaculum butyriciproducens} & Species & 1.66\\
12 & \textit{Parabacteroides} & Genus & 1.65\\
13 & Tannerellaceae & Family & 1.54\\
14 & \textit{Eubacterium rectale} & Species & 0.86\\
15 & \textit{Roseburia intestinalis} & Species & 0.85\\
16 & \textit{Roseburia} & Genus & 0.80\\
17 & Actinobacteria & Phylum & 0.75\\
18 & \textit{Parabacteroides merdae} & Species & 0.25\\
19 & Actinobacteria & Class & 0.20\\
20 & \textit{Oscillibacter sp\_57\_20} & Species & 0.18\\
21 & \textit{Roseburia sp\_CAG\_471} & Species & 0.18\\
22 & \textit{Faecalibacterium prausnitzii} & Species & 0.17\\
23 & \textit{Bacteroides} & Genus & 0.14\\
24 & Bacteroidaceae & Family & 0.14\\
25 & \textit{Roseburia faecis} & Species & 0.10\\
26 & Bifidobacteriaceae & Family & 0.10\\
27 & Bifidobacteriales & Order & 0.06\\
28 & \textit{Bifidobacterium} & Genus & 0.05\\
29 & \textit{Ruminococcus} & Genus & 0.03\\
30 & \textit{Ruminococcus bromii} & Species & 0.02\\
31 & \textit{Alistipes finegoldii} & Species & 0.01\\
32 & \textit{Alistipes putredinis} & Species & 0.01\\
33 & \textit{Coprococcus} & Genus & <0.01\\
34 & Lactobacillales & Order & <0.01\\
35 & \textit{Barnesiella intestinihominis} & Species & <0.01\\
36 & \textit{Blautia obeum} & Species & <0.01\\
37 & \textit{Bacteroides caccae} & Species & <0.01\\
38 & \textit{Coprococcus comes} & Species & <0.01\\
39 & \textit{Collinsella aerofaciens} & Species & <0.01\\
40 & \textit{Streptococcus} & Genus & <0.01\\
41 & \textit{Parabacteroides distasonis} & Species & <0.01\\
\bottomrule
\end{tabular}
\end{singlespace}
\endgroup{}
\end{table} 
\end{minipage}
\end{table}

% \input{appendix_arxiv.tex}
%This won't compile for local run but on arxiv I've flattened out the folder so it needs to be like this with no  references either.

\end{document}